\documentclass[a4paper]{article}
 
\usepackage[a4paper, total={6in, 10in}]{geometry}
 
 \usepackage{authblk}
\usepackage{makeidx}  % allows for indexgeneration
\usepackage{amsmath}
\usepackage{mathpartir}
\usepackage{xcolor}

\usepackage{amsthm}

\usepackage{chngcntr}
\usepackage{apptools}
\AtAppendix{\counterwithin{lemma}{section}}

\newcommand{\N}{{\cal N}}
\newcommand{\NA}{a}
\newcommand{\NB}{b}
\newcommand{\NC}{c}
\newcommand{\ND}{d}
\newcommand{\NE}{e}
\newcommand{\NX}{x}
\newcommand{\NY}{y}
\newcommand{\NZ}{z}
\newcommand{\fn}[1]{\mathsf{fn}(#1)}
\newcommand{\bn}[1]{\mathsf{bn}(#1)}
\newcommand{\n}[1]{\mathsf{n}(#1)}

\newcommand{\sym}[1]{\mathsf{sym}(#1)}
\newcommand{\ttype}{\omega}
\newcommand{\nub}{{\mathbf{\nu}}}

\newcommand{\NR}{{\mathbf{r}}}
\newcommand{\NS}{{\mathbf{s}}}
\newcommand{\NT}{{\mathbf{t}}}

\newcommand{\role}{\rho}

\newcommand{\varx}{x}

\newcommand\dual[1]{\overline{#1}}

\newcommand{\PP}{P}
\newcommand{\PQ}{Q}
\newcommand{\PR}{R}

\newcommand{\inact}{0}
\newcommand{\parop}{\;|\;}
\newcommand{\rest}[1]{(\nu #1)}
\newcommand{\scope}[1]{(#1)}
\newcommand{\prefix}{\sigma}

\newcommand{\msg}{l}
\newcommand{\send}[4]{#1!#4}

\newcommand{\receive}[4]{#1?#4}

\newcommand{\repreceive}[4]{!\scope #1 #1?#4}
\newcommand{\sauth}[4]{#1\langle#4\rangle}

\newcommand{\rauth}[4]{#1(#4)}

\newcommand{\red}{\rightarrow}

\newcommand{\lts}[1]{\xrightarrow{#1}}
\newcommand{\subst}[2]{\{#1/#2\}}

\newcommand\context{\mathcal{C}}

\newcommand{\operator}{\mathit {drift}}

\newcommand{\ophole}[3]{\operator(#1;#2;#3)}
\newcommand{\opholes}[5]{\operator(#1;#2;#3;#4;#5)}
\newcommand{\optop}[3]{\operator(#1;#2;#3)}

\newcommand{\rulename}[1]{\text{\small \textsc{#1}}}

\newenvironment{theorema}[2]{\begin{trivlist}
\item[\hskip \labelsep {\bfseries Theorem #1}] {#2}}{\end{trivlist}}

\newenvironment{lemmata}[2]{\begin{trivlist}
\item[\hskip \labelsep {\bfseries Lemma #1}] {#2}}{\end{trivlist}}

\newenvironment{corollar}[2]{\begin{trivlist}
\item[\hskip \labelsep {\bfseries Corollary #1}] {#2}}{\end{trivlist}}

\newenvironment{prop}[2]{\begin{trivlist}
\item[\hskip \labelsep {\bfseries Proposition #1}] {#2}}{\end{trivlist}}

\title{A Calculus for Modeling Floating Authorizations}
\author[1]{Jovanka~Pantovi\'c} 
\author[1]{Ivan~Proki\'c} 
\author[2]{Hugo {Torres Vieira}} %\inst{3}}

\affil[1]{Faculty of Technical Sciences, University of Novi Sad, Serbia}
\affil[2]{IMT School for Advanced Studies Lucca, Lucca, Italy}

\date{}

\newtheorem{definition}{Definition}[section]

\newtheorem{proposition}{Proposition}[section]
\newtheorem{theorem}{Theorem}[section]
\newtheorem{lemma}{Lemma}[section]
\newtheorem{corollary}{Corollary}[section]

\begin{document}

\maketitle              % typeset the title of the contribution

\begin{abstract}
Controlling resource usage in distributed systems is a challenging task 
given the dynamics involved in access granting. Consider, for instance, the 
setting of floating licenses where access can be granted if the request 
originates in a licensed domain and the number of active users is within the 
license limits, and where licenses can be interchanged. 
Access granting in such scenarios is given in terms of floating authorizations, 
addressed in this paper as first class entities of a process calculus 
model, encompassing the notions of domain, 
accounting and delegation. 
%Our development  focuses on communication-centred systems, building on the 
%$\pi$-calculus, where the only resources are communication channels. 
We 
present the operational semantics of the model in two equivalent alternative ways, 
each informing on the specific nature of authorizations. We also introduce a typing 
discipline to single out systems that never get stuck due to lacking authorizations, addressing configurations where authorization assignment is not statically prescribed in the system specification. %\comment{review}
%The abstract should summarize the contents of the paper
%using at least 70 and at most 150 words. 
%\keywords{}
\end{abstract}
%

%!TEX root =  auth.tex

\section{Introduction}

Despite the continuous increase of computational resources, their usage will 
always nevertheless be subject to availability, not to mention accessibility. Regardless whether such 
resources are of hardware or software nature, they might have finite or 
virtually infinite capabilities. Physical examples, that can be mapped 
to finite capabilities directly, include actual devices such as printers 
or cell phones, and components of a computing system such as memory or 
processors. Virtual examples such as a shared memory cell or a web service 
can be more easily seen as having infinite potential but often their availability 
is also finitely constrained.
In general, ensuring proper resource usage is a crucial yet non trivial task, 
given the highly dynamic nature of access requests and the flexibility necessary
to handle such requests while ensuring secure and efficient system operation.

In particular for security purposes it is crucial to control access to resources 
so as to guarantee, on the one hand, that access is granted only to authorized 
users and, on the the other hand, that granting access is subject to availability. 
Concrete examples include a wireless access point that has a determined policy 
to grant access and a limited capacity on the number of connected devices, and
a software application that is licensed to be used internally to an institution in a 
bounded way. Both examples include a limited amount of capabilities that are 
accessible in a shared way to a number of authorized users, involving a notion 
of \emph{floating} authorizations, borrowing the terminology from licensing. 
Essentially, the authorization to use the resource is associated to a capacity
and access can be granted to any authorized user only up to the point the 
capacity is reached, so floating captures the flexible nature of the access 
granting. We also identify a notion of \emph{implicit} granting since users may
be granted access under a certain \emph{domain} (e.g., the licensed institution) 
and not even be aware of the fact a capacity bound is present, at least up to the 
point access is denied.

Further pursuing the licensing setting we find examples which hint on the 
dynamic dimension of authorizations. Consider a user that deploys a software 
application on the cloud, for which the licensing may be provided by the user 
himself, a notion sometimes dubbed \emph{Bring Your Own License}. 
Such scenario involves a notion of 
authorization \emph{delegation}, in which the user may actually lose access 
to use the application locally given the capacity constraints. We identify here 
a notion of \emph{explicit} authorization granting, where the intention to 
yield and obtain an authorization are somehow signaled. 
%Under such flexibility and dynamic requirements, controlling resource 
%accesses are authorized is a challenging central task for what concerns 
%proper resource usage.
%
We therefore distill the following notions in a floating authorizations model: 
{\bf domain} so as to specify where access may be (implicitly) granted; 
{\bf accounting} so as to capture capacity; and {\bf delegation} so as to model 
(explicit) authorization granting. 

In this paper we present a model that 
encompasses these notions, developed focusing on a process 
calculus tailored for communication centered systems. In our model the only 
resources considered are communication channels, so our exploration on 
authorization control is carried out considering authorizations refer to channels 
and their usage (for communication) is controlled. 
Our development builds on the $\pi$-calculus~\cite{pi_calculus}, which 
provides the basis to model communicating systems including name passing, 
and on a model for authorizations~\cite{clar:eke}, from where we adopt 
the language constructs for authorization domain \emph{scope} and delegation. 
We adapt here the interpretation for authorization domains so as to encompass 
accounting, which thus allows us to model floating authorizations. To the best of 
our knowledge, ours is the first process calculus model that addresses floating
resources (in our case authorizations) as first class entities.

After presenting our language we also show a typing discipline that ensures
systems never incur in authorization errors, i.e., systems that are blocked 
due to lacking authorizations. Our type analysis addresses systems for which 
the authorization assignment is not statically prescribed in the system specification, 
given the particular combination between $\pi$-calculus name passing and floating authorizations.
In particular when authorizations for (some) received names are already held by the receiving
parties, a notion we call contextual authorizations.
We start by informally introducing our model by means of examples.

%!TEX root =  auth.tex 

\subsection{Examples} 

In this section we present our language by showing some examples
resorting to the licensing setting for the sake of a more intuitive reading.
First of all, to model authorization domains we consider a scoping 
operator, and we may write $\scope{\mathit{license}} \mathit{University}$ to 
represent that the $\mathit{University}$ domain holds one $\mathit{license}$.
This means that one use of $\mathit{license}$ within $\mathit{University}$ 
is authorized. In particular, if $\mathit{University}$ comprises two students that 
are simultaneously active, which we dub $\mathit{Alice}$ and $\mathit{Bob}$, 
we may write 
$$\scope{\mathit{license}} (\mathit{Alice} \parop \mathit{Bob})$$
in which case either $\mathit{Alice}$ or $\mathit{Bob}$ can use the 
$\mathit{license}$ that is floating, but not both of them. 
The idea is to support a notion of accounting, so when one of the students uses the
license the scope is confined accordingly. For example, if $\mathit{Bob}$ 
uses the license and evolves to $\mathit{LicensedBob}$ then the system
evolves to 
$$ \mathit{Alice} \parop \scope{\mathit{license}}\mathit{LincensedBob}$$ 
where
the change in the scope denotes that $\mathit{license}$ is not available
to $\mathit{Alice}$ anymore. 
%\rep{Notice the idea is that $\mathit{Bob}$ grabs
%the authorization directly by using it so as to model implicit granting.}
The evolution of the system is such that
the authorization is directly confined to $\mathit{Bob}$, who is using it, so as to model implicit granting.
Hence, at this point $\mathit{Alice}$ cannot implicitly grab the authorization 
and gets stuck if she tries to use $\mathit{license}$. Hence, 
name $\mathit{license}$ may be shared between $\mathit{Alice}$ and
$\mathit{Bob}$, so the (change of) scoping does not mean the name is privately held by 
$\mathit{LicensedBob}$, just the authorization.

Now consider a more generous 
university  
$$\scope{\mathit{license}}\scope{\mathit{license}} (\mathit{Alice} \parop \mathit{Bob}) \parop \mathit{Carol}$$
which specifies two authorizations (for the same resource) are available to
 $\mathit{Alice}$ and $\mathit{Bob}$, however none is specified for 
$\mathit{Carol}$. As mentioned, $\mathit{Carol}$ cannot use the 
authorization implicitly, but she can explicitly ask for it. To model explicit
authorization granting we introduce two communication primitives that allow 
for an authorization to be delegated. We write
$\sauth{\mathit{auth}}\role\msg{\mathit{license}}.\mathit{UnlicensedBob}$
to represent the explicit delegation of one authorization for $\mathit{license}$
via communication channel $\mathit{auth}$, after which activating configuration 
$\mathit{UnlicensedBob}$. Instead, by
$\rauth{\mathit{auth}}\role\msg{\mathit{license}}.\mathit{LicensedCarol}$ we 
represent the dual primitive that allows to receive one authorization for
$\mathit{license}$ via channel $\mathit{auth}$ after which activating 
configuration $\mathit{LicensedCarol}$. So by
$$\scope{\mathit{license}} \scope{\mathit{auth}}\sauth{\mathit{auth}}\role\msg{\mathit{license}}.\mathit{UnlicensedBob} \parop \scope{\mathit{auth}}\rauth{\mathit{auth}}\role\msg{\mathit{license}}.\mathit{LicensedCarol}$$
we represent a system where the authorization for $\mathit{license}$ can
be transferred from the delegating user to the receiving user leading to
$$ \scope{\mathit{auth}}\mathit{UnlicensedBob} \parop \scope{\mathit{auth}}\scope{\mathit{license}}\mathit{LicensedCarol}$$
where the scope of the authorization for $\mathit{license}$ changes accordingly. 
The underlying 
communication is carried out by a synchronization on channel $\mathit{auth}$,
for which we remark the respective authorizations $\scope{\mathit{auth}}{}$ are
present. 
%\del{{, so as to represent an authentication of both sides}}. 
In fact, in our model the only resources considered are communication
channels and their (immediate) usage is subject to the (implicit) authorization 
granting mechanism. 

Our model supports a form of fairness and does not allow a ``greedy'' usage of resources. For example, in 
$$\scope{\mathit{license}} (\mathit{Alice} \parop \scope{\mathit{license}}\mathit{LincensedBob})$$
the user $\mathit{LincensedBob}$ is considered granted by the closest $\scope{\mathit{license}}$ and cannot be confined with the floating one, except in case he needs both licenses (possibly for delegation).

We remark that no name passing is involved in the authorization delegation 
mechanism and that name $\mathit{license}$ is known to both ends in the first 
place. Instead, name passing is supported by dedicated primitives, namely
$$\scope{\mathit{comm}}\send{\mathit{comm}}\role\msg{\mathit{license}}.{\mathit{Alice}}
\parop \scope{\mathit{comm}}\receive{\mathit{comm}}\role\msg{\mathit{x}}.{\mathit{Dylan}}$$
represents a system where the name $\mathit{license}$ can be passed
from the left hand side to the right hand side via a synchronization on 
channel $\mathit{comm}$, leading to the activation of $\mathit{Alice}$ 
and $\mathit{Dylan}$ where the placeholder $x$ is instantiated by
$\mathit{license}$. Notice that the synchronization can take place since
the authorizations to use channel $\mathit{comm}$ are given, one for
each endpoint. 

Name passing allows to model systems where access to channels 
changes dynamically (since the 
communicated names refer to channels) but, as hinted in the previous 
examples, knowing a name does not mean being authorized to use it. 
So for instance 
$\scope{\mathit{comm}}\receive{\mathit{comm}}\role\msg{\mathit{x}}.\send{\mathit{x}}\role\msg{\mathit{reply}}.\inact$ specifies a reception
in $\mathit{comm}$ where the received name is then used to output 
$\mathit{reply}$, leading to an inactive state, represented by $\inact$. 
Receiving $\mathit{license}$ in the initial duly authorized use of
channel $\mathit{comm}$ leads to 
$\scope{\mathit{comm}}\send{\mathit{license}}\role\msg{\mathit{reply}}.\inact$
where the authorization for $\mathit{comm}$ is present but no authorization 
for $\mathit{license}$ is acquired as a result of the communication. Hence
the output specified using $\mathit{license}$ is not authorized and cannot take
place. We remark that an authorization for $\mathit{reply}$ is not required, 
hence communicating a name does not mean using it for the purpose of
authorizations. By separating name passing and authorization delegation we 
are then able to model systems where unauthorized intermediaries (e.g., brokers)
may be involved in forwarding names between authorized parties without ever 
being authorized to use them, for example $\scope{\mathit{comm}}\receive{\mathit{comm}}\role\msg{\mathit{x}}.\scope{\mathit{forward}}\send{\mathit{forward}}\role\msg{\mathit{x}}.\inact.$ 

So, after receiving a name
%\del{, except for the case when the authorization for the received name is granted implicitly,} 
a typical pattern can then be 
$$\scope{\mathit{comm}}\receive{\mathit{comm}}\role\msg{\mathit{x}}.
\scope{\mathit{auth}}\rauth{\mathit{auth}}\role\msg{\mathit{x}}.
\mathit{LicensedDylan}$$
where after the (authorized) reception on $\mathit{comm}$, an authorization
reception via (authorized) $\mathit{auth}$ is specified, upon which the 
authorization to use the received name is acquired. Another possibility 
for enabling authorizations for received names is to use the 
authorization scoping construct, e.g., 
$$\scope{\mathit{comm}}\receive{\mathit{comm}}\role\msg{\mathit{x}}.
\scope{\mathit{x}}\mathit{LicensedDylan}$$
where the authorization $\scope{x}$ is instantiated by the received name.
This example hints on the fact that the authorization scoping is a powerful 
mechanism that may therefore be reserved to the Trusted Computing Base, while 
resorting to the more controllable authorization delegation mechanism in general.

To introduce the last constructs of our language, consider the system
$$
\repreceive{{\mathit{license}}}\role\msg{x}.\scope{x}\sauth{\mathit{license}}\role\msg{x}.\inact
\parop \rest{\mathit{fresh}} \scope{\mathit{license}} 
\send{\mathit{license}}\role\msg{\mathit{fresh}}.
\rauth{\mathit{license}}\role\msg{\mathit{fresh}}. \inact
$$
where a licensing server is specified on the left hand side, used in the
specification given on the right hand side. By $\rest{\mathit{fresh}}
\mathit{Domain}$ 
we represent the creation of a new name which is private to 
$\mathit{Domain}$, so on the specification on the right hand side can be read 
as first create a name, then send it via (authorized) channel 
$\mathit{license}$, after which receive the authorization to 
use the $\mathit{fresh}$ name via channel $\mathit{license}$
and then terminate (the received authorization is not actually 
used in this simple example). On the left hand side we find a 
replicated (i.e., repeatably available) reception on (authorized)
channel $\mathit{license}$, after which an authorization for
the received name is specified that may be delegated via $\mathit{license}$.
In two communication steps the authorization for the newly created name
can therefore be transferred.

Returning to our university setting, consider example 
$$
\scope{exam}\scope{minitest}
\scope{\mathit{alice}}\receive{\mathit{alice}}\role\msg{x}.
\receive{x}\role\msg{\mathtt{Task}}. \mathit{DoTask}
%\parop 
%\scope{\mathit{bob}}\receive{\mathit{bob}}\role\msg{x}.
%\receive{x}\role\msg{\mathtt{Task}}. \mathit{DoTask})
$$
where two authorizations are available,
one for $\mathit{exam}$ and another for $\mathit{minitest}$,
and where $\mathit{alice}$ is waiting to receive the channel on 
which her assessment will be made.
Assuming that she can only take the $\mathit{exam}$ or the 
$\mathit{minitest}$ the authorizations specified are sufficient to 
carry out the reception of the $\mathtt{Task}$ (assuming an extension 
of the language considering other values which use is not
subject to authorizations). 
Which authorization will actually be used depends
on the received name, so the authorization is implicitly taken
directly when using the received channel. Naturally, if a name 
$\mathit{viva}$ is sent to the student then the prescribed
authorizations do not suffice. 

In order to capture the fact that the above configuration is safe, provided it is 
inserted in a context that matches the assumptions described previously, 
our types identify the names that can be safely communicated. 
For instance we may say that only names $\mathit{exam}$ and $\mathit{minitest}$ 
can be communicated in channel $\mathit{alice}$. Also, consider that $\mathit{alice}$ 
is a name not subject to instantiation and that $\mathit{exam}$ and $\mathit{minitest}$
can only receive values that are not subject to authorization control. We denote by
$\{\mathit{alice}\}( \{\mathit{exam},\mathit{minitest}\} ( \emptyset ))$ the type of channel
$\mathit{alice}$ in such circumstances, i.e., when it is not subject to replacement (we 
will return to this point),
that it can be used to communicate $\mathit{exam}$ and $\mathit{minitest}$ that in turn
cannot be used for communication (typed with $\emptyset$), reading from left to right. Using this 
information we
can ensure that the specification given for $\mathit{alice}$ above is safe, since all names
that will possibly be used in communications are authorized. 

To analyse the use of the input variable $\mathit{x}$ we then take into account that it can be 
instantiated by either $\mathit{exam}$ or $\mathit{minitest}$ (which cannot be used for
channel communication) so the type of $x$ is $\{\mathit{exam},\mathit{minitest}\} ( \emptyset )$.
Hence the need to talk about possible replacements of a name, allowing us to uniformly 
address names that are bound in inputs. Our types for channels are then built out of two parts, one 
addressing possible replacements of the channel identity itself ($\ttype$), and the other 
informing on the (type of the) names that may be exchanged in the channel ($T$), denoted
by $\ttype(T)$.

The typing assumption 
$\mathit{alice}: \{\mathit{alice}\}( \{\mathit{exam},\mathit{minitest}\} ( \emptyset ))$
informs on the possible contexts where the system above can be safely used. For instance
it is safe to compose with the system
$\scope{\mathit{alice}}\send{\mathit{alice}}\role\msg{minitest}$
%\scope{\mathit{minitest}}\send{minitest}\role\msg{\mathtt{Task}}$ 
where $\mathit{minitest}$ is sent to $\mathit{alice}$, since the name to be sent
belongs to the names expected on $\mathit{alice}$.
Instead, consider configuration 
$$
\scope{exam}\scope{minitest}
(\scope{\mathit{alice}}\receive{\mathit{alice}}\role\msg{x}.
\receive{x}\role\msg{\mathtt{Task}}. \mathit{DoTask}
\parop 
\scope{\mathit{bob}}\receive{\mathit{bob}}\role\msg{x}.
\receive{x}\role\msg{\mathtt{Task}}. \mathit{DoTask})
$$
which is also safe and addressed by our typing analysis considering
typing assumptions 
$\mathit{alice}: \{\mathit{alice}\}( \{\mathit{exam}\} ( \emptyset ))$ and
$\mathit{bob}: \{\mathit{bob}\}( \{\mathit{minitest}\} ( \emptyset ))$. Notice that which
authorization is needed by each student
is not statically specified in the system, which is safe when
both $\mathit{exam}$ and $\mathit{minitest}$ are sent given the authorization
scopes can be confined accordingly. 
Clearly, the typing specification already informs on the association and a symmetric
association is also admissible.

The typing analysis shown in Section~\ref{sec:Types} addresses 
such configurations where authorizations for received names may be 
provided by the context. In Section~\ref{sec:Calculus} we present
the operational semantics of our language considering two equivalent 
alternatives that inform on the specific nature of authorizations in our model.

%!TEX root =  auth.tex

\section{A Model of Floating Authorizations}\label{sec:Calculus}

In this section we present our process model, an extension of the $\pi$-calculus~\cite{pi_calculus} with
specialized constructs regarding authorizations adopted from a model for authorizations~\cite{clar:eke}. 
\begin{table}[t]
\begin{math}
\displaystyle
\begin{array}[t]{@{}rcl@{\quad}l@{}}
  \PP,\PQ & ::= & \inact & \text{(Inaction)} \\
          &  & \PP\parop\PQ & \text{(Parallel)}\\
          &  & \rest\NA \PP & \text{(Restriction)}\\
          &  & \send\NA\role\msg\NB.\PP & \text{(Output)} \\
          &  & \receive\NA\role\msg{x}.\PP & \text{(Input)}\\
\end{array}
\qquad \quad
\begin{array}[t]{@{}rcl@{\quad}l@{}}
            &    & \scope{\NA} \PP & \text{(Authorization)}\\
            &    & \sauth\NA\role\msg\NB.\PP & \text{(Send authorization)}\\
            &    & \rauth\NA\role\msg\NB.\PP & \text{(Receive authorization)}\\
            &    & \repreceive\NA\role\msg{x}.\PP & \text{(Replicated input)}\\
\end{array}
\end{math}
\caption{\label{tab:syntax}Syntax of processes.}
\end{table}
The syntax of the language is given in Table~\ref{tab:syntax}. It assumes a countable set of {\it names} $\N$, ranged over by $\NA, \NB, \NC, \ldots, \NX, \NY, \NZ, \ldots$ We briefly present the syntactic constructs adopted from the $\pi$-calculus. An inactive process is represented by $\inact$. $\PP\parop\PQ$ represents two processes simultaneously active, that may interact via synchronization in channels. $\rest\NA \PP$ is the name restriction construct, specifying the creation of a channel name $\NA$ that is known only to the process $\PP.$ The output prefixed process $\send\NA\role\msg\NB.\PP$ sends the name $\NB$ on channel $\NA$ and proceeds as $\PP,$ and the input prefixed process $\receive\NA\role\msg{x}.\PP$ receives on channel $\NA$ a name and substitutes the placeholder $\NX$ in $\PP$ with the received name.
We comment on the remaining constructs introduced to model authorizations in more detail:
\begin{itemize}
\item The term $\scope{\NA} \PP$ is another scoping mechanism for names, representing that process $\PP$ has one authorization to use channel $\NA.$
In contrast to the name restriction, name $\NA$ is not private to $\PP$.
\item The term $\sauth\NA\role\msg\NB.\PP$ represents the process that delegates one authorization for the name $\NB$ along the name $\NA$ and proceeds as $\PP.$ 
\item The term $\rauth\NA\role\msg\NB.\PP$ represents the process which receives one authorization for the name $\NB$ along the name $\NA$ and proceeds as $\PP.$ 
%\comment{H: Move the following sentence} Since $\NB$ is not bound in $\PP$, constructs send and receive authorization can only affect the possible change of the scope of authorization $\scope{\NB},$ and can not be used for  name passing.
\item The term $\repreceive\NA\role\msg{x}.\PP$ allows us to specify infinite behavior: the process receives the name along the (authorized) name $\NA$ and substitutes the placeholder $\NX$ in $\PP$ with the received name, activating in parallel a copy of the original process.
%$\repreceive\NA\role\msg{x}.\PP.$
\end{itemize} 

In $\rest\NX \PP,$ $\receive\NA\role\msg{x}.\PP$ and $\repreceive\NA\role\msg{x}.\PP$ the name $\NX$ is {\it binding} with scope $\PP.$ All occurrences of a name that are binding, or that are under the scope of it binding occurrence, are said to be {\it bound}. If the occurrence of the name is not bound in a term, it is said to be {\it free}. We use $\fn\PP$ and $\bn\PP$ to denote the sets of free and bound names in $\PP,$ respectively. 
Regarding language constructs for authorization manipulation, we have that in $\scope\NA\PP$  occurrence of the name $\NA$ is free and occurrences of names $\NA$ and $\NB$ in processes $\sauth\NA\role\msg\NB.\PP$ and $\rauth\NA\role\msg\NB.\PP$ are also free. We remark that in our model authorization scope extrusion is not applicable since a free name is specified, unlike name restriction (see Table~\ref{tab:structural}), and constructs to send and receive authorizations can only affect the possible changes of the scope of authorization $\scope{\NB},$ and do not involve name passing.

\subsection{Reduction Semantics}

As in the $\pi$-calculus, the essence of the behavior of processes can be seen as communication. Specific to our model is that two processes ready to synchronize on a channel must be authorized to use the channel. For example, $\scope\NA\send\NA\role\msg\NB.\PP \parop \scope\NA\receive\NA\role\msg\NX.\PQ$ can evolve to $\scope\NA\PP \parop \scope\NA\PQ\subst{\NB}{\NX},$ since both sending and receiving actions are authorized, while $\scope\NA\send\NA\role\msg\NB.\PP \parop \receive\NA\role\msg\NX.\PQ$ lacks the proper authorization on the receiving end, hence the synchronization cannot occur. Another specific aspect of our language is authorization delegation. For example, consider $\scope\NA\scope\NB\sauth\NA\role\msg\NB.\PP \parop \scope\NA\rauth\NA\role\msg\NB.\PQ.$ 
%\del{where the process on the left hand side is authorized on both $\NA$ and $\NB$.} 
Both actions along name $\NA$ are authorized, and the delegating process has the respective authorization on $\NB$, hence the authorization delegation can take place, leading to $\scope\NA\PP \parop \scope\NA\scope\NB\PQ$. Notice that the authorization  for $\NB$ changed to scope over to the process that received the authorization. If actions along name $\NA$ are not authorized or the process delegating authorization for $\NB$ is not authorized on $\NB$, like in $\scope\NA\sauth\NA\role\msg\NB.\PP \parop \scope\NA\rauth\NA\role\msg\NB.\PQ,$ then the synchronization is not possible.

We formally define the behavior of processes by means of a reduction semantics and afterwards by means of a labeled transition system.
\begin{table}[t]
\[
\begin{array}{@{}c@{}} 
  \inferrule[(sc-par-inact)]{}
  {\PP\parop\inact\equiv\PP}\qquad\quad
  \inferrule[(sc-par-comm)]{}
  {\PP\parop\PQ\equiv\PQ\parop\PP}\qquad\quad
  \inferrule[(sc-par-assoc)]{}
  {(\PP\parop\PQ)\parop\PR\equiv\PP\parop(\PQ\parop\PR)}
\vspace{2mm}\\
  {\inferrule[(sc-res-inact)]{}
  {\rest\NA\inact\equiv\inact}
  \qquad\quad}
  \inferrule[(sc-res-swap)]{}
  {\rest\NA\rest\NB\PP\equiv\rest\NB\rest\NA\PP}
      \qquad\quad
  \inferrule[(sc-rep)]{}
  {\repreceive\NA\role\msg{x}.\PP\equiv\; \repreceive\NA\role\msg{x}.\PP \parop 
  \scope{\NA}\receive\NA\role\msg{x}.\PP  }      
\vspace{2mm}\\
  \inferrule[(sc-res-extr)]{}
  {\PP\parop\rest\NA\PQ\equiv\rest\NA(\PP\parop\PQ)\quad (\NA\notin\fn\PP)}
    \qquad\quad
  \inferrule[(sc-alpha)]{}
  {\PP \equiv_{\alpha} \PQ \implies \PP \equiv \PQ}
\vspace{2mm} \\
  \inferrule[(sc-auth-swap)]{}
  {\scope{\NA}\scope{\NB}\PP \equiv \scope{\NB}\scope{\NA}\PP}
    \qquad
  \inferrule[(sc-auth-inact)]{}
  {\scope{\NA}\inact \equiv \inact}
\qquad
  \inferrule[(sc-scope-auth)]{}
  {\scope{\NA}\rest{\NB}\PP \equiv \rest{\NB}\scope{\NA}\PP \quad\text{if $\NA\neq \NB$}
  }
%\vspace{2mm}\\
%	\new{\inferrule[(sc-rep-in)]{}
%	{\repreceive\NA\role\msg{x}.\PP \equiv \scope\NA\receive\NA\role\msg{x}.\PP \parop \repreceive\NA\role\msg{x}.\PP}}
	
%  \qquad
%  \inferrule[]{}
%  {\PP \equiv_{\alpha} \PQ \implies \PP \equiv \PQ}
\end{array}
\]
%\vspace{-4mm}
\caption{\label{tab:structural}Structural congruence.}
\end{table}
%The operational semantic for a language usually could be given in a two ways, one describing possible computations within individual processes, called {\it reduction semantics}, and the other describing possible interactions between process and process environment, called {\it an action semantics}. This section presents the reduction semantic for our language.
%
Reduction is defined as a binary relation between processes, denoted $\red$, where $\PP \red \PQ$ specifies that process $\PP$ evolves to process $\PQ$ in one computational step. In order to identify processes which differ syntactically but have the same behavior, we introduce the {\it structural congruence} relation $\equiv$, which is the least congruence relation between processes satisfying the rules given in Table~\ref{tab:structural}. Most rules are standard considering structural congruence in the $\pi$-calculus. In addition we adopt some rules introduced previously~\cite{clar:eke} that manipulate authorization scoping, namely \rulename{(sc-auth-swap)}, \rulename{(sc-auth-inact)}, 
and \rulename{(sc-scope-auth)} . 
 
%\comment{H: I don't think the rule for replication is a novelty, it is a particular instance of the standard $\pi$-calculus one $!P \equiv P | !P$} \new{The only novelty is the rule \rulename{(sc-rep-in)} expressing the infinite behavior of replicated input.} 
Regarding authorization scoping, we remark there is no rule which relates authorization scoping and parallel composition, like \rulename{(sc-res-extr)}  for name restriction. This is due to the interpretation of authorization scoping, as adopting a rule of the sort $\scope{\NA}(\PP \parop \PQ)\equiv \scope{\NA}\PP\parop\scope{\NA}\PQ$ would represent introducing/discarding one authorization, thus interfering with authorization accounting. Hence we distinguish $\scope{\NA}(\PP \parop \PQ)$ where the authorization is shared between $\PP$ and $\PQ$ and $\scope{\NA}\PP\parop\scope{\NA}\PQ$ where two authorizations are specified, one for each process. Another approach could be a rule of the sort $\scope{\NA}(\PP \parop \PQ)\equiv \PP\parop \scope{\NA} \PQ$, which also may affect the computational power of a process. For example, two processes $\send\NA\role\msg\NB.\inact \parop \scope{\NA}\inact$ and $\scope{\NA}(\send\NA\role\msg\NB.\inact \parop  \inact)$ should not be considered equal since the first one is not authorized to perform the output, while the second one is. In contrast, notice that in $\scope{\NA}(\send\NA\role\msg\NB.\PP \parop \PQ)$ the output on channel $a$ is authorized, but if the action is carried out then the authorization is \emph{confined}
to $\PP$ and no longer available for $\PQ$, since one authorization can only be 
(effectively) used by a single thread. 

\begin{table}[t]
\[
\displaystyle
\begin{array}[t]{@{}r@{\quad}c@{\quad}l@{}}
  \context[\cdot] & ::= & \cdot \quad | \quad \PP\parop \context[\cdot] \quad
          | \quad \scope{\NA} \context[\cdot]\\
  \context[\cdot_1, \cdot_2] & ::= & \context[\cdot_1]\parop\context[\cdot_2]
\quad \parop \quad   \PP\parop \context[\cdot_1, \cdot_2] \quad
          | \quad \scope{\NA} \context[\cdot_1, \cdot_2] 
\end{array}
\]
\caption{\label{tab:Contexts} Contexts with one and two holes.}
\end{table}

\begin{table}[t]
\[
\begin{array}{c} %@{\quad}l}
\inferrule[(c-end)]{} 
{\ophole{\cdot}{\emptyset}{\tilde \ND}=\cdot}
\qquad
\inferrule[(c-rem)]{ \ophole{\context[\cdot]}{\tilde{\NA}}{ \tilde \ND, \NC}=\context'[\cdot]
} 
{\ophole{\scope\NC\context[\cdot]}{\tilde{\NA},\NC\, }{ \tilde \ND}=\context'[\cdot]}
\\\\
\inferrule[(c-skip)]{ \ophole{\context[\cdot]}{\tilde{\NA}}{ \tilde \ND  }=\context'[\cdot] \quad \NC \not \in \tilde \ND} 
{\ophole{\scope\NC\context[\cdot]}{\tilde{\NA}}{ \tilde \ND}=\scope\NC \context'[\cdot]}
\qquad
\inferrule[(c-par)]{ \ophole{\context[\cdot]}{\tilde{\NA}}{\tilde \ND}=\context'[\cdot]}  
{\ophole{\context[\cdot] \parop \PR}{\tilde{\NA}}{\tilde \ND}=\context'[\cdot] \parop \PR}
\\\\
\end{array}
\]
\caption{\label{tab:op_for_one_hole}Definition of ${\operator}$ on contexts with one hole}
\end{table}

\begin{table}[t]
\[
\begin{array}{c} %@{\quad}l}
 \inferrule[(c2-spl)]{
\ophole{\context_1[\cdot_1]}{ \tilde{\NA}}{ \tilde \ND} = \context_1'[\cdot] \quad 
\ophole{\context_2[\cdot_2]}{ \tilde{\NB}}{ \tilde \NE} = \context_2'[\cdot]} 
{
\opholes{\context_1[\cdot_1]\parop\context_2[\cdot_2]}{\tilde{\NA}}{\tilde{\NB}}{
\tilde \ND} { \tilde \NE}= \context_1'[\cdot_1]\parop\context_2'[\cdot_2]}
\\\\
\inferrule[(c2-rem-l)]{
\opholes{\context[\cdot_1,\cdot_2]}{ \tilde{\NA}}{ \tilde{\NB}}{ \tilde \ND, \NC}{ \tilde \NE} = \context'[\cdot_1,\cdot_2] } 
{
\opholes{\scope\NC\context[\cdot_1,\cdot_2]}{\tilde{\NA},\NC\,}{\tilde{\NB}}{\tilde \ND}{ \tilde{\NE}}= \context'[\cdot_1,\cdot_2]}
\qquad
\inferrule[(c2-rem-r)]{
\opholes{\context[\cdot_1,\cdot_2]}{ \tilde{\NA}}{ \tilde{\NB}}{ \tilde \ND}{ \tilde \NE, \NC} = \context'[\cdot_1,\cdot_2]
} 
{
\opholes{\scope\NC\context[\cdot_1,\cdot_2]}{\tilde{\NA}}{\tilde{\NB},\NC\, }{ \tilde \ND}{ \tilde{\NE}}= \context'[\cdot_1,\cdot_2]}
 \\\\
\inferrule[(c2-skip)]{
\opholes{\context[\cdot_1,\cdot_2]}{ \tilde{\NA}}{\tilde{\NB}}{\tilde \ND}{ \tilde{\NE}} = \context'[\cdot_1,\cdot_2] \quad \NC \not \in \tilde \ND, \tilde{\NE}} 
{
\opholes{\scope\NC\context[\cdot_1,\cdot_2]}{\tilde{\NA}}{\tilde{\NB}}{\tilde \ND}{ \tilde{\NE}}= \scope\NC\context'[\cdot_1,\cdot_2]}
 \qquad
 \inferrule[(c2-par)]{
\opholes{\context[\cdot_1,\cdot_2]}{ \tilde{\NA}}{\tilde{\NB}}{\tilde \ND}{ \tilde{\NE}} = \context'[\cdot_1,\cdot_2]} 
{
\opholes{\context[\cdot_1,\cdot_2]\parop\PR}{\tilde{\NA}}{ \tilde{\NB}}{ \tilde \ND}{ \tilde{\NE}}= \context'[\cdot_1,\cdot_2]\parop\PR}
 \\\\
\end{array}
\]
\caption{\label{tab:op_for_two_holes}Definition of ${\operator}$ on contexts with two holes}
\end{table}

Structural congruence is therefore not expressive enough to isolate two authorized processes willing to communicate on the same channel. For example process $\scope\NA(\scope\NA(\PQ \parop \send\NA\role\msg\NB.0) \parop \receive\NA\role\msg\NX.0)$ 
cannot be rewritten, using structural congruence rules, to  $\PQ\parop\scope\NA\send\NA\role\msg\NB.0 \parop \scope\NA\receive\NA\role\msg\NX.0$. However both processes are able to reduce to process $\PQ\parop\scope\NA\inact \parop \scope\NA \inact$
since the actions are under the scope of the proper authorizations.
To define the reduction relation we thus introduce an auxiliary notion of static {\it  contexts} with one and two holes and an operation that allow us to single out the 
configurations where communication can occur. Intuitively, if two processes have active prefixes ready for synchronization, and both are under the scope of the appropriate authorizations, then the reduction step is possible. 

Static contexts are defined in Table~\ref{tab:Contexts} following standard lines. We use $\cdot_1$ and $\cdot_2$ notation to avoid ambiguity (i.e., when
$ \context[\cdot_1, \cdot_2] = \context[\cdot_1]\parop\context[\cdot_2]$ then
$ \context[\PP, \PQ] = \context[\PP]\parop\context[\PQ]$).
Note that in Table~\ref{tab:Contexts} there is no case for name restriction construct $\rest\NA,$ which allows to identify specific names and avoid unintended name capture. Remaining cases specify holes can occur in parallel composition and underneath the authorization scope, the only other contexts underneath which processes are deemed active. We omit the symmetric cases for parallel 
composition since contexts will be considered up to structural congruence.

Operation $\operator$ plays a double role: on the one hand it is defined only when the hole/holes is/are under the scope of the appropriate number of authorizations in the context; on the other hand, when defined, it yields a context obtained from the original one by removing specific authorizations (so as to capture \emph{confinement}). In our model, the specific authorizations that are removed for the sake of confinement are the ones nearest to the occurrence of the hole.

The operator $\ophole{\context[\cdot]}{\tilde{\NA}}{ \tilde \ND }$ is defined inductively by the derivation rules shown in Table~\ref{tab:op_for_one_hole}. We present the rules reading from the conclusion to the premise. The operator takes as arguments a context with one hole and  lists of names $\tilde{\NA}$ and $\tilde \ND,$ in which the same name can appear more than once.  The first list of names represents the names of authorizations that are to be removed from the context and the second represents the names of authorizations that have already been removed by the same operation. 

We briefly comment on the rules shown in Table~\ref{tab:op_for_one_hole}.
The rule where the authorization is removed from the context \rulename{(c-rem)} 
specifies that the (same) name is passed from the first list to the second list,
hence from ``to be removed'' to ``has been removed''.
In the rule where the authorization is preserved in the context \rulename{(c-skip)} we check if the name specified in the authorization is not on the second list (has been removed), hence only authorizations that were not already removed proceeding towards the hole can be preserved. This ensures the removed authorizations  are the ones nearest to the hole. The rule for parallel composition \rulename{(c-par)} is straightforward and the base rule \rulename{(c-end)} is defined only if the first list is empty. This implies that the operator is defined only when all authorizations from the first list are actually removed from the context up to the point the hole is reached. It should be noted that the second list of names is only for internal use to the operation, so, top level, when defining the operator for some context $\context[\cdot]$ and some list of names $\tilde{\NA}$ that are to be removed from the context, no authorizations have been removed and the respective list is empty.
 For example,
\[
\begin{array}{l}
  \ophole{ \scope\NA \cdot}{ \NA }{ \emptyset }= \cdot                     \\
  \ophole{ \scope\NA (\scope\NA \cdot \parop R) }{ \NA }{ \emptyset }= \scope\NA (\cdot \parop R)  \\
   \ophole{ \cdot }{ \NA }{ \emptyset } \text{  is undefined }  \\
  \ophole{ \scope\NA \scope\NB \cdot}{ \NA,\NB } { \emptyset } = \cdot                    \\
  \ophole{ \scope\NA  \cdot}{ \NA,\NB } { \emptyset } \text{ is undefined }   \\
  \ophole{ \scope\NA \cdot \parop \scope \NB \inact }{ \NA,\NB } { \emptyset } \text{ is undefined.}
\end{array}
\]

For the sake of defining reduction, where a pair of interacting processes must be identified, we require a generalization of the operation to contexts with two holes. The operator $\opholes{\context[\cdot_1,\cdot_2]}{\tilde{\NA}}{ \tilde{\NB}}{ \tilde \ND}{ \tilde{\NE}}$, defined inductively by the rules shown in Table~\ref{tab:op_for_two_holes}, takes as arguments a context with two holes, two lists of names $\tilde {\NA}$ and $\tilde {\NB}$ representing the names of authorizations which are to be removed and two list of names $ \tilde \ND$ and $ \tilde{\NE}$ representing names of authorizations already removed. Lists 
$\tilde {\NA}$ and $\tilde{\ND}$ refer to the $\cdot_1$ hole while $\tilde{\NB}$ and $\tilde{\NE}$ refer to the $\cdot_2$ hole. 

We briefly describe the rules reading from conclusion to premise. 
Rule \rulename{(c2-spl)} describes the case for two contexts with one hole each, in which case the respective operation (for one-hole contexts) is used to obtain the resulting context, considering the name lists $\tilde \NA$ and $\tilde{\ND}$ for the context on the left hand side and $\tilde{\NB}$ and $\tilde{\NE}$ for the context on the right hand side. The remaining rules follow exactly the same lines of the ones shown in Table~\ref{tab:op_for_one_hole}, duplicating authorization removal so as to address the two pairs of lists in a dedicated way.
For example,
\[
  \begin{array}{l}
    \opholes{\scope\NB \scope \NA \scope \NA ( \cdot_1 \parop \cdot_2) }{ \NA,\NB}{\NA}{\emptyset}{\emptyset} = \cdot_1 \parop \cdot_2 \\
    \opholes{\scope\NB \scope \NA ( \cdot_1 \parop \scope\NA \cdot_2) }{ \NA,\NB}{\NA}{\emptyset}{\emptyset} = \cdot \parop \cdot \\
    \opholes{\scope\NA\scope\NB  \cdot_1 \parop \scope\NA \cdot_2 }{ \NA,\NB }{\NA}{\emptyset}{\emptyset} = \cdot_1 \parop \cdot_2 \\
    \opholes{\scope\NB ( \cdot_1 \parop \scope\NA \scope\NA \cdot_2)}{ \NA,\NB}{\NA}{\emptyset}{\emptyset} \text{ is undefined.} \\
  \end{array}
\]

Notice that the operation carried out for contexts with two holes relies at some point on the operators for contexts with one hole and the fact that the derivation is possible only if the axioms for empty contexts are true. Thus, the operator is undefined if the proper authorizations are lacking. As before, lists $\tilde{\ND}$ and $\tilde{\NE}$ are used only internally by the operator. For the rest of presentation we abbreviate $\opholes{\context[\cdot_1,\cdot_2]}{\tilde{\NA}}{ \tilde{\NB}}{ \emptyset}{ \emptyset}$ with $\optop{\context[\cdot_1,\cdot_2]}{\tilde{\NA}}{ \tilde{\NB}}$.

\begin{table}[t]
\[
\begin{array}[t]{@{}c@{\qquad}c@{}}
{\inferrule[(r-comm)]{\optop{\context[\cdot_1,\cdot_2]}{\NA}{\NA} = \context'[\cdot_1 ,\cdot_2 ]}
{\context [\send\NA\role\msg\NB.\PP,\receive\NA\role\msg{x}.\PQ]
\red
\context'[\scope{\NA}\PP,\scope{\NA}\PQ\subst{\NB}{x}]}
}
%\multicolumn{2}{c}
%{\inferrule[(r-rep-comm)]{\new{\context^-[\cdot_1,\cdot]=\operator(\context 
%[\cdot,\cdot]; \NA;\emptyset)}}
%{\context [\send\NA\role\msg\NB.\PP,\repreceive\NA\role\msg{x}.\PQ]
%\red
%\context^-[\scope{\NA}\PP, \scope{\NA}\PQ\subst{\NB}{x} \parop 
%\repreceive\NA\role\msg{x}.\PQ])
%}}
%\\\\
&
{\inferrule[(r-auth)]{
\optop{\context[\cdot_1,\cdot_2]}{\NA,\NB}{\NA} = \context'[\cdot_1 ,\cdot_2 ]}
%\new{\context^-[\cdot,\cdot]=\operator(\context [\cdot ,\cdot ]; \NA,\NB;\NA)}}
{\context [\sauth\NA\role\msg\NB.\PP,\rauth\NA\role\msg\NB.\PQ]
\red
\context' [\scope{\NA}\PP, \scope{\NA}\scope{\NB}\PQ)]
}}
\\\\
\inferrule[(r-stru)]
{\PP \equiv \PP' \red \PQ' \equiv \PQ}
{\PP \red \PQ}
&
%\inferrule[(r-parc)]
%{\PP \red \PQ}
%{\PP \parop \PR \red \PQ \parop \PR}
%&
\inferrule[(r-newc)]
{\PP \red \PQ}
{\rest\NA \PP \red \rest\NA \PQ}
%&
%\inferrule[(r-autc)]
%{\PP \red \PQ}
%{\scope\NA \PP \red \scope\NA \PQ}
\end{array}
\]
\caption{\label{tab:Reduction}Reduction rules.}
\end{table}
We may now present the reduction rules, shown in Table \ref{tab:Reduction}.
Rule \rulename{(r-comm)} states that two processes can synchronize on name $\NA,$ passing name $\NB$ from emitter to receiver, only if both processes are under the scope of, at least one per each process, authorizations for name ${\NA}$. The yielded process considers the context where the two authorizations have been removed by the $\operator$ operation, and specifies the \emph{confined} authorizations for $\NA$ which scope only over the continuations of the communication prefixes $\PP$ and $\PQ$.
%In rule \rulename{(r-rep-comm)} at least one authorization is required only for the process willing to output, replicated input is authorized by the definition. After the reduction another copy of the  process with replicated input is created in parallel, and the operator $\operator$ only deletes one appearance of authorization $\scope{\NA},$ closest to the output prefix.
Analogously to \rulename{(r-comm)}, rule \rulename{(r-auth)} states that two process can exchange authorization $\scope{\NB}$ on a name $\NA$ only if the first process is under the scope of, at least one, authorization $\NB$ and, again, if both processes are authorized to perform an action on name $\NA$. 
As before, the yielded process considers the context where the authorizations 
have been removed by the $\operator$ operation. Notice that the authorization for $\NB$ is removed for the delegating process and confined to the receiving process so as to model the exchange.
Finally, the rule \rulename{(r-stru)} closes reduction under structural congruence, and rule  \rulename{(r-newc)} closes reduction under the restriction construct $\rest\NA$. 
Note there are no rules that close reduction under parallel composition and authorization scoping, as these constructs are already addressed by the static contexts in \rulename{(r-comm)} and \rulename{(r-auth)}. There is also no rule dedicated to replicated input since, thanks to structural congruence rule \rulename{(sc-rep)}, a single copy of replicated process may be distinguished and take a part in a synchronization captured by \rulename{(r-comm)}.

%To illustrate the operator $^{-}$, recall the process and the its context representation
%\[
%\PP=\scope{\NA}(\scope{\NA}\scope{\NA}(\scope{\NB}\sauth\NA\role\msg\NB.\inact\parop\rauth\NA\role\msg\NB.\inact)\parop \PR)=\context _{\NA^2\NB;\NA} [\sauth\NA\role\msg\NB.\inact,\rauth\NA\role\msg\NB.\inact].
%\]
%Using the rule \rulename{(r-auth)} it reduces to 
%\[
%\PQ=\scope{\NA}((\scope{\NA}\inact\parop\scope{\NA}\scope{\NB}\inact)\parop \PR)=\context^{-} _{\NA^2\NB;\NA} [\scope{\NA}\inact,\scope{\NA}\scope{\NB}\inact],
%\] 
%where $\context^{-} _{\NA^2\NB;\NA} [\cdot_1,\cdot_2]=\scope{\NA}((\scope{\NA}[\cdot_1]\parop\scope{\NA}\scope{\NB}[\cdot_2] )\parop \PR),$ i.e. only counted authorizations with the closest appearance to the holes are removed. 

To illustrate the rules and the $\operator$ operation, consider process 
\[
\PP=\scope\NA(\scope\NA(\PQ \parop \send\NA\role\msg\NB.\PR_1) \parop \receive\NA\role\msg\NX.\PR_2).
\]
where we may say that $\PP=\context[\send\NA\role\msg\NB.\PR_1, \receive\NA\role\msg\NX.\PR_2]$ and
$
\context[\cdot_1, \cdot_2]= \scope\NA(\scope\NA(\PQ \parop \cdot_1) \parop \cdot_2)
$.
Applying $\operator$  to the context $\context[\cdot_1, \cdot_2]$ to remove the two authorizations for name $\NA$ we have 
$
\optop{\context [\cdot_1 ,\cdot_2 ]}{ \NA}{\NA} = (\PQ \parop \cdot_1) \parop \cdot_2.
$
Thus, $\PP\red (\PQ\parop \scope\NA \PR_1) \parop \scope\NA\PR_2\subst{\NB}{\NX}$, where we observe that the authorizations are confined to the continuations $\PR_1$ and $\PR_2\subst{\NB}{\NX}$.

Synchronizations in our model are tightly coupled with the notion of authorization, in the sense that in the absence of the proper authorizations the synchronizations cannot take place. We characterize such undesired configurations, referred to as \emph{error} processes, by identifying the 
redexes singled-out in the reduction semantics which are stuck due to the lack of the necessary authorizations. Roughly, this is the case when the premise of the reduction rules is not valid, hence when the $\operator$ operation is not defined. 
% Informally, process is not an authorization error if all of its active prefixes are under the scope of appropriate authorizations. A prefix in a process is active if it can be placed within an {\it active context}.  Active context $\ucontext [\cdot]$ is defined in a usual way
%\[
%\ucontext[\cdot] \quad ::= \quad \cdot \quad | \quad \PP\parop \ucontext[\cdot] \quad
%          | \quad \scope{\NA} \ucontext[\cdot] \quad \parop \quad \rest\NA\ucontext[\cdot].
%\]

We introduce some abbreviations useful for the remaining presentation: a prefix $\alpha_\NA$ stands for any communication prefix along name $\NA$, i.e. $\send\NA\role\msg\NB, \receive\NA\role\msg\NX, \sauth\NA\role\msg\NB$ or $\rauth\NA\role\msg\NB$ and $\rest{\tilde{\NA}}$ stands for $\rest{\NA_1}\ldots\rest{\NA_n}$ when $\tilde\NA =\NA_1,\ldots, \NA_n$.

\begin{definition}[Error]\label{d:error}
Process $\PP$ is an error if $P\equiv \rest{\tilde\NC} \context[\alpha_\NA.\PQ, \alpha_\NA'.\PR]$ and 
\begin{enumerate}
\item 
$\alpha_\NA = \send\NA\role\msg\NB$, 
$\alpha_\NA' = \receive\NA\role\msg\NX$
and
$ \optop{\context[\cdot_1,\cdot_2]}{\NA}{\NA}$  is undefined, or 
\item 
$\alpha_\NA = \sauth\NA\role\msg\NB$, 
$\alpha_\NA' = \rauth\NA\role\msg\NB$
and
$ \optop{\context[\cdot_1,\cdot_2]}{\NA,\NB}{\NA}$ is undefined.
\end{enumerate}
\end{definition}

Notice that the definition of errors is aligned with that of reduction, where structural congruence is used to identify a configuration (possibly in the scope of a number of restrictions) that directly matches one of the redexes given for reduction, but where the respective application of $\operator$ is undefined.
The type analysis presented afterwards singles out processes that never incur in   errors, but first we show an alternative characterization of the operational semantics.

\subsection{Action Semantics}

In this section we introduce a labeled transition system (LTS) that provides an equivalent (as shown later) alternative representation of the operational semantics of our model. As usual the LTS is less compact, albeit more informative, with respect to the reduction semantics. The basic notion is that of observable actions, ranged over by $\alpha$, which are identified as follows: 
\[
\alpha ::= \scope\NA^i\send\NA\role\msg\NB \; \parop \; \scope\NA^i\receive\NA\role\msg\NB \; \parop \; \scope\NA^i\scope\NB^j\sauth\NA\role\msg\NB \; \parop \; \scope\NA^i\rauth\NA\role\msg\NB \; \parop \;  \rest\NB\scope\NA^i\send\NA\role\msg\NB \; \parop \; \tau_\omega
\]
where $\omega$ is of the form 
$\scope\NA^{i+j}\scope\NB^k$ and $i,j,k\in\{0,1\}$. 
We may recognize the communication action prefixes together with some annotations that capture carried/lacking authorizations and bound names. Intuitively, a communication action tagged with 
$\scope{\NA}^0$ represents the action is not carrying an authorization on $\NA$,
while $\scope{\NA}^1$ represents the action is carrying an authorization on $\NA$. Notice in the case for authorization delegation two such annotations are present, one for each name involved. As usual $\rest \NB$ is used to denote the name in the object of the communication is bound (cf. $\pi$-calculus bound output). In the case of internal steps, the $\omega$ annotation identifies the authorizations lacking for the synchronization to take place, and we use $\tau$ to abbreviate $\tau_{\scope\NA^{0}\scope\NB^0}$ where no authorizations are lacking.
For a given label $\alpha$, $\n\alpha,$ $\fn\alpha$ and $\bn\alpha$ represents the set of label names, free names and bound names, respectively, all defined in expected lines. 

\begin{table}[t]
\[
\begin{array}[t]{@{}c@{}}
\inferrule[(l-out)]{}
{\send\NA\role\msg\NB.\PP\lts{\send\NA\role\msg\NB}\scope{\NA}\PP}
\quad
\inferrule[(l-in)]{}
{\receive\NA\role\msg\varx.\PP\lts{\receive\NA\role\msg\NB}\scope{\NA}\PP\subst\NB\varx}
\quad
\inferrule[(l-out-a)]{}
{\sauth\NA\role\msg\NB.\PP\lts{\sauth\NA\role\msg\NB}\scope{\NA}\PP}
\quad
\inferrule[(l-in-a)]{}
{\rauth\NA\role\msg\NB.\PP\lts{\rauth\NA\role\msg\NB}\scope{\NA}\scope{\NB}\PP}
\vspace{3mm}\\
\inferrule[(l-in-rep)]{}
{\repreceive\NA\role\msg{x}.\PP\lts{\scope{\NA}\receive\NA\role\msg\NB}\scope{\NA}\PP\subst\NB\varx\parop\repreceive\NA\role\msg{x}.\PP}
\qquad
\inferrule[(l-par)]
{\PP\lts{\alpha}\PQ \quad \bn\alpha\cap\fn\PR=\emptyset}
{\PP\parop\PR\lts{\alpha}\PQ\parop\PR}
\vspace{3mm}\\
\inferrule[(l-res)]
{\PP\lts{\alpha}\PQ \quad \NA\notin n(\alpha)}
{\rest\NA\PP\lts{\alpha}\rest\NA\PQ}\qquad
\inferrule[(l-open)]
{\PP\lts{\scope\NA^i\send\NA\role\msg\NB}\PQ \quad a\not= b}
{\rest\NB\PP\lts{\rest\NB\scope\NA^i\send\NA\role\msg\NB}\PQ }
\vspace{3mm}\\
\inferrule[(l-scope-int)]
{\PP\lts{\tau_{\omega\scope{\NA}}}\PQ}
{\scope\NA\PP\lts{\tau_\omega}\PQ}
\qquad
\inferrule[(l-scope-ext)]
{\PP\lts{\prefix_\NA}\PQ }
{\scope\NA\PP\lts{\scope\NA\prefix_\NA}\PQ}
\qquad
\inferrule[(l-scope)]
{\PP\lts{\alpha}\PQ \quad \tau_{\omega\scope{\NA}}\neq \alpha\neq \prefix_\NA}
{\scope\NA\PP\lts{\alpha}\scope{\NA}\PQ}
\vspace{3mm}\\
\inferrule[(l-comm)]
{\PP\lts{\scope\NA^i\send\NA\role\msg\NB}\PP' \quad \PQ\lts{\scope\NA^j\receive\NA\role\msg\NB}\PQ' \quad \omega=\scope\NA^{2-i-j}}
{\PP\parop\PQ\lts{\tau_\omega}\PP'\parop\PQ'}
\vspace{3mm}\\
\inferrule[(l-close)]
{\PP\lts{\rest\NA\scope\NB^i\send\NB\role\msg\NA}\PP' \quad \PQ\lts{\scope\NB^j\receive\NB\role\msg\NA}\PQ' \quad \omega=\scope\NB^{2-i-j}\quad \NA\notin fn(\PQ)}
{\PP\parop\PQ\lts{\tau_\omega}(\nu \NA)(\PP'\parop\PQ')}
\vspace{3mm}\\
\inferrule[(l-auth)]
{\PP\lts{\scope\NB^k\scope\NA^i\sauth\NA\role\msg\NB}\PP' \quad \PQ\lts{\scope\NA^j\rauth\NA\role\msg\NB}\PQ' \quad \omega=\scope\NA^{2-i-j}\scope\NB^{1-k}}
{\PP\parop\PQ\lts{\tau_\omega}\PP'\parop\PQ'}
\end{array}
\]
\caption{\label{tab:Transition}The transition rules.}
\end{table}

The transition relation is the least relation included in ${\cal P} \times {\cal A} \times {\cal P}$, where $\cal P$ is the set of all processes and $\cal A$ is the set of all actions, that satisfies the rules in Table~\ref{tab:Transition}, which we now briefly describe.
The rules \rulename{(l-out)}, \rulename{(l-in)}, \rulename{(l-out-a)}, \rulename{(l-in-a)} capture the actions that correspond to the communication prefixes. Notice that in each rule the continuation is activated under the scope of the authorization required for the action (and provided in case of authorization reception), so as to capture confinement. Notice the labels are not decorated with the corresponding authorizations, which represents that the actions are not carrying any authorizations, omitting $\scope{\NA}^0$ annotations.
In contrast, replicated input is authorized by construction, which is why in \rulename{(l-in-rep)} the label is decorated with the corresponding authorization. 
Rule \rulename{(l-par)} is adopted from the $\pi$-calculus, lifting the actions of one of the branches (the symmetric rule is omitted) while avoiding unintended name capture. 

The rules for restriction \rulename{(l-res)} and \rulename{(l-open)} follow the lines of the ones given for the $\pi$-calculus. Rule \rulename{(l-res)} says that actions of $\PP$ are also actions of $\rest\NA\PP,$ provided that the restricted name is not specified in the action, and \rulename{(l-open)} captures the bound output case, opening the scope of the restricted name $\NA$, thus allowing for scope extrusion. 
The rule \rulename{(l-scope-int)} shows the case of a synchronization that lacks an authorization on $\NA$, so at the level of the authorization scope the action exhibited no longer lacks the respective authorization and leads to a state which no longer specifies the authorization scope. We use ${\omega\scope\NA}$ to abbreviate $\scope\NA^{2}\scope\NB^k$, $\scope\NA^{1}\scope\NB^k$, and $\scope\NB^{i+j}\scope\NA^1$ in which case $\omega$ is obtained by the respective exponent decrement. We remark that in contrast to the extrusion of a restricted name via bound output, where the scope floats \emph{up} to the point a synchronization (rule \rulename{(l-close)} explained below), authorization scopes actually float \emph{down} to the level of communication prefixes (cf. rules \rulename{(l-out)}, \rulename{(l-in)}, \rulename{(l-out-a)}, \rulename{(l-in-a)}), so as to capture confinement. Rule \rulename{(l-scope-ext)} follows similar lines as \rulename{(l-scope-int)} as it also refers to lacking authorizations, specifically for the case of an (external) action that is not carrying a necessary authorization. We use $\prefix_\NA$ to denote both an action that specifies $\NA$ as communication subject (cf. $\alpha_\NA$) and is annotated with $\scope\NA^0$ (including bound output), and of the form
$\scope\NB^i\sauth\NB\role\msg\NA$ where $i \in \{0,1\}$ (which includes $ \scope{\NA}^1\sauth\NA\role\msg\NA$ where a second authorization on $\NA$ is lacking). We also use $\scope{\NA}\prefix_\NA$ to denote the respective annotation exponent increase. Rule \rulename{(l-scope)} captures the case of an action that is not lacking an authorization on $\NA$, in which case the action crosses seamlessly the authorization scope for $\NA$.

The synchronization of parallel processes is expressed by the last three rules, omitting the symmetric cases. In rule \rulename{(l-comm)} one process is able to send and other to receive a name $\NB$ along name $\NA$ so the synchronization may take place. Notice that if the sending and receiving actions are not carrying the appropriate authorizations, then the transition label $\tau_\omega$ specifies the lacking authorizations (the needed two minus the existing ones).
In rule \rulename{(l-close)} the scope of a bound name is closed. One process is able to send a bound name $\NA$ and the other to receive it, along name $\NB$, so the synchronization may occur leading to a configuration where the restriction scope is specified (avoiding unintended name capture) so as to finalize the scope extrusion. 
The authorization delegation is expressed by rule \rulename{(l-auth)}, where an extra annotation for $\omega$ is considered given the required authorization for the delegated authorization. Carried authorization annotations, considered here up to permutation, thus identify,  in a compositional way, the requirements for a synchronization to occur.

To illustrate the transition system, let us consider process
\[
\PP=\scope{\NA}(\scope{\NA}\scope{\NA}(\scope{\NB}\sauth\NA\role\msg\NB.\inact\parop\rauth\NA\role\msg\NB.\inact)\parop \PR).
\]
Using the rule \rulename{(l-out-a)} we obtain $\sauth\NA\role\msg\NB.\inact\lts{\sauth\NA\role\msg\NB}\scope{\NA}0$ and using \rulename{(l-scope-ext)} we get $\scope{\NB}\sauth\NA\role\msg\NB.\inact\lts{\scope{\NB}\sauth\NA\role\msg\NB}\scope{\NA}0.$ In parallel, by rule \rulename{(l-in-a)} $\rauth\NA\role\msg\NB.\inact\lts{\rauth\NA\role\msg\NB}\scope{\NA}\scope{\NB}0,$ and using the rule \rulename{(l-auth)} we obtain 
\[
\scope{\NB}\sauth\NA\role\msg\NB.\inact\parop\rauth\NA\role\msg\NB.\inact\lts{\tau_{\scope{\NA}\scope{\NA}}}\scope{\NA}0\parop\scope{\NA}\scope{\NB}0.
\] 
Since the action is pending on two authorizations $\scope{\NA}$ we now apply \rulename{(l-scope-int)} twice
\[
\scope{\NA}\scope{\NA}(\scope{\NB}\sauth\NA\role\msg\NB.\inact\parop\rauth\NA\role\msg\NB.\inact)\lts{\tau}\scope{\NA}0\parop\scope{\NA}\scope{\NB}0.
\] 
Applying \rulename{(l-par)} and \rulename{(l-scope)} we get

\[
\PP\lts{\tau}\scope{\NA}((\scope{\NA}0\parop\scope{\NA}\scope{\NB}0)\parop\PR).
\]

For the sake of showing the equivalence between the semantics induced by reduction and by the labelled transition system we must focus on $\tau$ transitions where no authorizations are lacking.

\begin{theorem}[Harmony]\label{cor:red=tau}
%\comment{H: Sometimes the notation $\PP \lts{\tau}\equiv Q$ is used in this circumstances, feel free to consider it} 
$\PP\red\PQ$ if and only if $\PP \lts{\tau}\equiv Q.$ 
\end{theorem}
\begin{proof}
We get one implication by induction on the derivation $\PP\red\PQ$ and the other by induction on the derivation on  $\PP\lts{\tau}\PQ$
(See Appendix~\ref{app:semproofs}).
\end{proof}

{Both presentations of the semantics inform on the particular nature of 
authorizations in our model. As usual, the labeled transition system is 
more directly explicit, but the more compact reduction semantics 
allows for a more global view of authorization manipulation.
In what follows we present the type system that allows to statically identify processes that never incur in authorization errors.}
%!TEX root =  auth.tex 

\section{Type System}\label{sec:Types}

In this section we present a typing discipline that allows to statically identify
\emph{safe} processes, i.e., that never incur in an error (cf.~Definition \ref{d:error}), hence
that do not exhibit actions lacking the proper authorizations. As mentioned in the Introduction, our typing analysis addresses configurations where authorizations can be granted contextually.
 %\comment{H: Maybe move the following to the introduction, if }
Before presenting the typing language which talks about the names that can be safely communicated on channels, we introduce auxiliary notions that cope with name generation, namely \emph{symbol} annotations and \emph{well-formedness}. 

Since the process model includes name restrictions and our types contain name identities, we require a symbolic handling of such bound names when they are included in type specifications. Without loss of generality, we refine the process model for the purpose of the type analysis adding an explicit symbolic representation of name restrictions. In this way we avoid a more involved treatment of bound names in typing environments. 

Formally, we introduce a countable set of symbols $\cal{S}$  ranged over by $\NR, \NS, \NT, \ldots,$ disjoint with the set of names $\N,$ and symbol $\nu$ not in $\N\cup\cal{S}.$ Also, in order to introduce a unique association of restricted names and symbols, we refine the syntax of the name creation construct $\rest \NA\PP$ in two possible ways: $\rest{\NA: \NR}\PP$ and $\rest{\NA: \nu}\PP,$ decorated with symbol from $\cal{S}$ or with symbol $\nu,$ respectively. We use $\sym\PP$ to denote a set of all symbols from $\cal{S}$ in process $\PP.$ Names associated with symbols from $\cal{S}$ may be subject to contextual authorizations, while names associated with symbol $\nu$ are not subject to contextual authorizations. The latter can be communicated in a more flexible way since on the receiver side there can be no expectation of relying on contextual authorizations.

For the purpose of this section we adopt the reduction semantics, adapted here considering refined definitions of structural congruence and reduction. In particular for structural congruence, we omit rule \rulename{(sc-res-inact)} 
$(\rest\NA\inact\equiv\inact)$
and we decorate name restriction accordingly in rules \rulename{(sc-res-swap)}, \rulename{(sc-res-extr)} and \rulename{(sc-scope-auth)}---e.g., $\PP\parop\rest{\NA: \NR}\PQ\equiv\rest{\NA: \NR}(\PP\parop\PQ)$ and $\PP\parop\rest{\NA: \nu}\PQ\equiv\rest{\NA: \nu}(\PP\parop\PQ)$ keeping the side condition $\NA\notin\fn\PP$. We remark that the omission of axiom \rulename{(sc-res-inact)} is not new in process models where name restriction is decorated with typing information (cf.~\cite{Boreale2010}).

The annotations with symbols from $\cal{S}$ allow to yield a unique identification of the respective restricted names. The processes we are interested in have unique occurrences of symbols from $\cal{S}$ and do not contain occurrences of such symbols in replicated input. We say that process $\PP$ is \emph{well-formed} if it does not contain two occurrences of the same symbol from $\cal{S}$ and for any subprocess $\PP'$ of $\PP$ that is prefixed by replicated input $\sym{\PP'}=\emptyset.$ As shown later, any typable process is well-formed, and we may show that well-formedness is preserved by structural congruence and reduction.

%\begin{proposition}[Preservation of Well-formedness]\label{lemm:well-formed}
%If $\PP$ is well-formed and $\PP\equiv\PQ$ or $\PP\red\PQ$ then $\PQ$ is also well-formed and $\sym{\PP}=\sym{\PQ}$.
%\end{proposition}
%\begin{proof}
%The proof is by induction on the size of the derivation structural congruence or reduction rule (See Appendix~\ref{app:typesproofs}). 
%\end{proof}

We may now introduce the type language, which syntax is given in Table~\ref{tab:syntax_type}, that allows to identify safe instantiations of channel names that are subject to contextual authorizations. By $\varphi$ we represent a set of names from $\N$ and of symbols from $\cal{S}$. We use $\ttype$, which stands for a set $\varphi$ or the symbol $\nu$, to characterize either names that may be instantiated (with a name in $\varphi$) or that are not subject to contextual authorizations ($\nu$). In a type $\ttype(T)$ the carried type $(T)$ characterizes the names that can be communicated in the channel. Type $\emptyset$ represents (ground) names that cannot be used for communication.
As usual, we define typing environments, denoted with $\Delta$, as a set of typing assumptions each associating a type to a name $ \NA: T$. We represent by $\mathit{names}(T)$ the set of names that occur in type $T$ and by $\mathit{names}(\Delta)$ the set of names that occur in all entries of $\Delta$.

\begin{table}[t]
$$ \ttype ::= \;  \varphi  \; \parop \; \nu
\qquad \qquad 
T ::= \; \ttype(T) \; \parop \; \emptyset
$$
\caption{Syntax of types}\label{tab:syntax_type}
\end{table}

%%%%%
We may now present the type system, defined inductively in the
structure of processes, by the rules given in Table \ref{tab:Typing rules}. A typing judgment
$
    \Delta\vdash_\rho \PP
$
%where $\rho$ denotes a multiset of names and $\rrec::= \rec \parop \nrec$. 
 states that $\PP$ uses channels as prescribed by typing environment $\Delta$ and that $\PP$ can only be placed in contexts that provide the authorizations given in $\rho$, which is a multiset of names (from $\N$, including their multiplicities).
The use of a multiset can be motivated by considering process 
 %$\sauth\NA\role\msg\NA.0$ and 
 $ \send\NA\role\msg\NB.0 \parop \receive\NA\role\msg\NX.0$ that can be typed as
 $
  % a: \{a\} (\{a\}) \vdash_{a,a } \sauth\NA\role\msg\NA.0 \qquad 
   a:\{a\}(\{b\}(\emptyset))\vdash_{\rho} \send\NA\role\msg\NB.0 \parop \receive\NA\role\msg\NX.0
 $
 where necessarily $\rho$ contains $\{a,a\}$ which specifies that the process can only be placed in contexts that offer two authorizations on name $\NA$ (one is required per each communicating prefix).

\begin{table}[t]
\[
\begin{array}[t]{@{}c@{\qquad}c@{\qquad}c@{}}

\inferrule[(t-stop)]
	{}
	{ \Delta\vdash_\rho \inact } 
&
\inferrule[(t-par)]
	{ \Delta\vdash_{\rho_1} \PP_1  \quad  \Delta\vdash_{\rho_2} \PP_2 \quad \sym{\PP_1}\cap\sym{\PP_2}=\emptyset }
	{ \Delta\vdash_{\rho_1\uplus\rho_2}  \PP_1 \parop \PP_2 }
&	
\inferrule[(t-auth)]
	{ \Delta\vdash_{\rho\uplus\{\NA\}} \PP }
	{ \Delta\vdash_{\rho} \scope{\NA}\PP }
\vspace{2mm}\\
\multicolumn{3}{c}
{
\inferrule[(t-new)]
	{ \Delta, \NA: \{\NA\}(T)\vdash_\rho \PP  \quad  \Delta'=\Delta\subst{\NR}{\NA} \quad \NR\notin\sym{\PP} \quad   \NA\notin\rho,\mathit{names} (T)}
	{ \Delta'\vdash_\rho \rest{\NA: \NR}\PP }
}
\vspace{2mm}\\
\multicolumn{3}{c}
{
\inferrule[(t-new-rep)]
	{ \Delta, \NA:\nub (T)\vdash_\rho \PP  \quad   \NA\notin\rho, \mathit{names} (T, \Delta) }
	{ \Delta\vdash_\rho \rest{\NA: \nub}\PP }
}
\vspace{2mm}\\
\multicolumn{3}{c}
{
\inferrule[(t-out)]
	{ \Delta\vdash_\rho \PP  \quad \Delta(\NA)=\ttype(\ttype'(T)) \quad \Delta(\NB)=\ttype''(T) 	 \quad  \ttype''\subseteq \ttype' \quad  \NA\notin\rho \Rightarrow \omega \subseteq\rho}
	{ \Delta\vdash_{\rho} \send\NA\role\msg\NB.\PP }
}
\vspace{2mm}\\
\multicolumn{3}{c}
{
\inferrule[(t-in)]
	{ \Delta,\NX:T\vdash_\rho \PP  \quad \Delta(\NA)=\omega(T) \quad \NX\notin\rho,\mathit{names}( \Delta) \quad \NA\notin\rho \Rightarrow \omega\subseteq\rho}
	{ \Delta\vdash_{\rho} \receive\NA\role\msg\NX.\PP }
}
\vspace{2mm}\\
\multicolumn{3}{c}
{
\inferrule[(t-rep-in)]
	{ \Delta,\NX:T\vdash_{\{\NA\}} \PP  \quad \Delta(\NA)=\ttype(T) \quad 
	  \NX \notin \rho, \mathit{names}(\Delta)\quad \sym{\PP}=\emptyset }
	{ \Delta\vdash_\rho\; \repreceive\NA\role\msg\NX.\PP }
}
\vspace{2mm}\\
\multicolumn{3}{c}
{
\inferrule[(t-deleg)]
	{ \Delta\vdash_\rho \PP \quad \Delta(\NA)=\ttype(T) \quad \NA\notin\rho \Rightarrow \ttype \subseteq\rho}
	{ \Delta\vdash_{\rho\uplus\{\NB\}} \sauth\NA\role\msg\NB.\PP }
}
\vspace{2mm}\\
\multicolumn{3}{c}
{
\inferrule[(t-recep)]
	{ \Delta\vdash_{\rho\uplus\{\NB\}} \PP  \quad \Delta(\NA)=\ttype(T) \quad \NA\notin\rho \Rightarrow \ttype\subseteq\rho}
	{ \Delta\vdash_{\rho} \rauth\NA\role\msg\NB.\PP }
}
\end{array}
\]
\caption{\label{tab:Typing rules} Typing rules.}
\end{table}

We comment on the salient points of the typing rules:
\begin{itemize}
\item \rulename{(t-stop)}:
 The inactive process is typable using any $\Delta$ and $\rho$.
\item  \rulename{(t-par)}:
 If processes $\PP_1$ and $\PP_2$ are typed under the same environment $\Delta,$ then $\PP_1 \parop \PP_2$ is typed under $\Delta$ also.  
 Consider that $\PP_1$ and $\PP_2$ own enough authorizations when placed in contexts providing authorizations $\rho_1$ and $\rho_2,$ respectively. Then, the process $\PP_1 \parop \PP_2$ will have enough authorizations
if it is placed in a context providing the sum of authorizations from $\rho_1$ and $\rho_2$. By $\rho_1\uplus\rho_2$ we represent the addition operation for multisets which sums the frequencies of the elements. The side condition $\sym{\PP_1}\cap\sym{\PP_2}=\emptyset$ is necessary to ensure the unique association of symbols and names.
\item  \rulename{(t-auth)}:
 Processes $\scope{\NA}\PP$ and $\PP$ are typed under the same environment $\Delta,$ due to the fact that scoping is a non-binding operator.
If process $\PP$ owns enough authorizations when placed in a context that provides authorizations $\rho\uplus \{\NA\},$ then $\scope{\NA}\PP$ owns enough authorizations when placed in a context  that provides authorizations $\rho$.
\item \rulename{(t-new)}:
 If process $\PP$ is typable under an environment that contains an entry for $\NA,$ then the process with restricted name $\rest{\NA: \NR}\PP$ is typed under the environment removing the entry for $\NA$ and where each occurrence of name $\NA$ in $\Delta$ is substituted by the symbol $\NR.$
By $\Delta\subst{\NR}{\NA}$ we represent the environment obtained by replacing every occurrence of $\NA$ by $\NR$ in every typing assumption in $\Delta$, hence in every type, where we exclude the case when $\NA$ has an entry in $\Delta$. The side condition $\NR\notin\sym\PP$ is necessary for the uniqueness of the symbol and name pairings and $\NA\notin\rho, \mathit{names}(T)$ says that the context cannot provide authorization for the private name and for ensuring consistency of the typing assumption.

The symbolic representation of a bound name in the typing environment enables us to avoid the case when a restricted (unforgeable) name could be sent to a process that expects to provide a contextual authorization for the received name. 
For example consider process $\rest{\NB: \NR}\scope\NA\send\NA\role\msg\NB.\inact \parop \scope\NA\scope\ND\receive\NA\role\msg\NX.\send\NX\role\msg\NC.\inact$ where a contextual authorization for $\ND$ is specified, a configuration excluded by our type analysis since the assumption for the type of channel $\NA$ carries a symbol (e.g., $\NA: \{\NA\}(\{\NR\}(\emptyset))$) for which no contextual authorizations can be provided. Notice that the typing of the process in the scope of the restriction uniformly handles the name, which leaves open the possibility of considering contextual authorizations for the name within the scope of the restriction.
 %If process is typed under assumption that it can be used as sub-term of a body of replicated input ($\rrec=\rec$) then $\NA:\rest\NA(T),$ ensuring that bound names prefixed with replicated input behaves as bound name even in the part of the process where it is free. (See Example~\ref{example_typable_rep_input} below.)
\item \rulename{(t-new-rep)}: The difference with respect to rule \rulename{(t-new)} is that no substitution is performed since the environment must already refer to symbol $\nu$ in whatever pertains to the restricted name (notice the side condition). For example to type process $\rest{\NB: \nu}\scope\NA\send\NA\role\msg\NB.\inact$ the type of $\NA$ must be $\ttype(\nu(T))$ for some $\ttype$ and $T$ where $\nu$ identifies the names communicated in $\NA$ are never subject to contextual authorizations.
\item  \rulename{(t-out)}:
  Process $\PP$ is typed under an environment 
where types of names $\NA$ and $\NB$ are such that all possible replacements for name $\NB$ (which are given by $\ttype''$) are safe to be communicated along name $\NA$ (which is formalized by $\ttype''\subseteq\ttype'$, where $\ttype'$ is the carried type of $\NA$), and also that $T$ (the carried type of $\NB$) matches the specification given in the type of $\NA$. In such case the process $\send\NA\role\msg\NB.\PP$ is typed under the same environment.
% $\Delta$ and it proceeds in the context providing the same authorizations, 
There are two possibilities to ensure name $\NA$ is authorized, namely the context may provide directly the authorization for name $\NA$ or it may provide authorizations for all replacements of name $\NA$, formalized as $\NA\notin\rho \Rightarrow \omega \subseteq\rho$. Notice this latter option is crucial to address contextual authorizations and that in such case $\omega$ does not contain symbols (since $\rho$ by definition does not).
\item \rulename{(t-in)}:
This rule is readable under the light of principles explained in the previous rule.
\item \rulename{(t-rep-in)}:
The continuation is typed considering an assumption for the input bound name $\NX$ and that $\rho=\{\NA\}$, which specifies that the expected context provides authorizations only for name $\NA$. In such a case, the replicated input is typable considering the environment obtained by removing the entry for $\NX$, which must match the carried type of $\NA$, provided that $\NX\notin\rho,\mathit{names}(\Delta)$ since it is bound to this process, and that process $\PP$ does not contain any symbols from $\cal{S}$, necessary to ensure the unique association of symbols and names when copies of the replicated process are  activated
 (see the discussion on the example shown in (\ref{example:error.rep}) at the end of this Section).  In that case, process $\repreceive\NA\role\msg\NX.\PP$ can be placed in any context that conforms with $\Delta$ and (any) $\rho$.
\item  \rulename{(t-deleg)}, \rulename{(t-recep)}: 
In these rules the typing environment is the same in premises and conclusion.
The handling of the subject of the communication ($\NA$) is similar  to, e.g., rule \rulename{(t-out)}. The way in which the authorization is addressed in rule \rulename{(t-recep)} follows the lines of rule \rulename{(t-auth)}. In rule \rulename{(t-deleg)} the authorization for $\NB$ is added to the ones expected from the context. Notice that in such way no contextual authorizations can be provided for delegation, but generalizing the rule is direct following the rules for other prefixes.
\end{itemize}

We say  process $\PP$ is well-typed if $\Delta\vdash_\emptyset \PP$ and $\Delta$ only contains assumptions of the form $\NA: \{\NA\}(T)$ or $\NA: \nu(T)$. At top level the typing assumptions address the free names of the process, which are not subject to instantiation. Free names are either characterized by $\NA: \{\NA\}(T)$ which says that $\NA$ cannot be replaced, or by $\NA: \nu(T)$ which says that $\NA$ is not subject to contextual authorizations. For example, process $\scope\NA\send\NA\role\msg\NB.\inact \parop \scope\NA\scope\NB\receive\NA\role\msg\NX.\send\NX\role\msg\NC.\inact$ is typable under the assumption that name $\NB$ has type $\{\NB\}(\{\NC\}(\emptyset))$, while it is not typable under the assumption $\nu(\{\NC\}(\emptyset))$.
%:\rrec.$ 
The fact that no authorizations are provided by the context ($\rho=\emptyset$) means that the process $\PP$ is self-sufficient in terms of authorizations.

We may now present our results, starting by mentioning some fundamental properties. We may show that typing derivations enjoy standard properties (Weakening and Strengthening) and that typing is preserved under structural congruence (Subject Congruence). 
As usual, to prove typing is preserved under reduction we need an auxiliary result that talks about name substitution.

\begin{lemma}[Substitution]\label{lemm:Substitution_lemma}
Let $ \Delta,\NX:\ttype(T)\vdash_\rho \PP$ and $\NX\notin\mathit{names}(\Delta)$. %Then 
\begin{itemize}
\item[1.] If $\Delta(\NA)=\{\NA\}(T)$  and $\NA\in \ttype$ then $\Delta\vdash_{\rho\subst{\NA}{\NX}} \PP\subst{\NA}{\NX}.$
\item[2.] If $\Delta(\NA)=\nub(T)$  and $\nub=\ttype$ then $\Delta\vdash_{\rho\subst{\NA}{\NX}} \PP\subst{\NA}{\NX}.$
\end{itemize} 
\end{lemma}
\begin{proof}
The proof is by induction on the depth of the derivation $ \Delta\vdash_\rho\PP$
(See Appendix~\ref{app:typesproofs}).  
\end{proof}

We remark that the name ($\NA$) must be contained in the set of possible instantiations for the name that $\NA$ is replacing ($\NX$), and the two names must have the same carried type ($T$). Even though subtyping is not present, the inclusion principle used in Lemma~\ref{lemm:Substitution_lemma} already hints on a substitutability notion. 

Our first main result says that typing is preserved under reduction. 

\begin{theorem}[Subject Reduction]\label{theorem:Subject_reduction}
If $\PP$ is well-typed, $ \Delta\vdash_\emptyset \PP$ and $\PP\red\PQ$ then $ \Delta\vdash_\emptyset \PQ.$

\end{theorem}

\begin{proof}
The proof is by case analysis on last reduction step (See Appendix~\ref{app:typesproofs}).
\end{proof}

Not surprisingly, since errors involve redexes, the proof of Theorem~\ref{theorem:Subject_reduction} is intertwined with the proof of the error absence property included in our second main result.
Proposition~\ref{lemm:Error_Freedom} captures 
the soundness of our typing analysis, i.e., that well-typed processes are well-formed and are not stuck due to the lack of proper authorizations, hence are not errors (cf. Definition~\ref{d:error}).

\begin{proposition}[Typing Soundness]\label{lemm:Error_Freedom}
\begin{itemize}
\item[1.] If $\Delta\vdash_\rho \PP$ then $\PP$ is well-formed.
\item[2.] If $\PP$ is well-typed then $\PP$ is not an error.
\end{itemize}
%\begin{itemize}
%\item[(i)] If $\PP\equiv \rest{\tilde\NC}\context[\send\NA\role\msg\NB.\PP_1, \receive\NA\role\msg\NX.\PP_2]$ and $ \Delta\vdash_\emptyset \PP$ where for each $\NC\in\mathit{dom}(\Delta),$ $\Delta(\NC)=\NC(T)$ or $\Delta(\NC)=\nub(T)$ then $ \optop{\context[\cdot_1,\cdot_2]}{\NA}{\NA}$ is defined and if $\context'[\cdot_1, \cdot_2]= \optop{\context[\cdot_1,\cdot_2]}{\NA}{\NA}$ and $\PQ\equiv \rest{\tilde\NC}\context'[\scope\NA\PP_1, \scope\NA\PP_2\subst{\NB}{\NX}]$ then $ \Delta\vdash_\emptyset \PQ.$
%\item[(ii)] If $\PP\equiv \rest{\tilde\NC}\context[\sauth\NA\role\msg\NB.\PP_1, \rauth\NA\role\msg\NB.\PP_2]$ and $ \Delta\vdash_\emptyset \PP$ where for each $\NC\in\mathit{dom}(\Delta),$ $\Delta(\NC)=\NC(T)$ or $\Delta(\NC)=\nub(T)$ then $ \optop{\context[\cdot_1,\cdot_2]}{\NA, \NB}{\NA}$ is defined and if $\context'[\cdot_1, \cdot_2]= \optop{\context[\cdot_1,\cdot_2]}{\NA, \NB}{\NA}$ and $\PQ\equiv \rest{\tilde\NC}\context'[\scope\NA\PP_1, \scope\NA\scope\NB\PP_2]$ then $ \Delta\vdash_\emptyset \PQ.$
%\end{itemize}
\end{proposition}
\begin{proof}
Immediate from auxiliary result (see Appendix~\ref{app:typesproofs}). 
\end{proof}

%\comment{H: see comment above}

%\begin{lemma}[Error Absence]\label{lemm:Error_Free}
%If $ \Delta\vdash_\emptyset \PP$ and for each $\NA\in\mathit{dom}(\Delta),$ $\Delta(\NA)=\NA(T)$ or $\Delta(\NA)=\nub(T)$ and $P\equiv \rest{\tilde\NC} \context[\alpha_\NA.\PQ, \alpha_\NA'.\PR]$ where 
%\begin{enumerate}
%\item 
%$\alpha_\NA = \send\NA\role\msg\NB$, 
%$\alpha_\NB = \receive\NA\role\msg\NX$
%then
%$ \optop{\context[\cdot_1,\cdot_2]}{\NA}{\NA}$  is defined, or 
%\item 
%$\alpha_\NA = \sauth\NA\role\msg\NB$, 
%$\alpha_\NB = \rauth\NA\role\msg\NB$
%then
%$ \optop{\context[\cdot_1,\cdot_2]}{\NA,\NB}{\NA}$ is defined.
%\end{enumerate}
%Therefore, $\PP$ is not an error.
%\end{lemma}

%\begin{proof}

%Directly from Lemma~\ref{lemm:Pre_Subject_reduction} and Lemma~\ref{lemm:Subject_congruence}. 
%\end{proof}
As usual, the combination of Theorem~\ref{theorem:Subject_reduction} and Proposition \ref{lemm:Error_Freedom} yields type safety.

\begin{corollary}[Type Safety]\label{cor:Type_Safety}
If $\PP$ is well-typed and $\PP\red^*\PQ$ then $\PQ$ is not an error.
\end{corollary}

The type safety ensures that well-typed processes will never lack the necessary authorizations to carry out their communications.  

To illustrate the typing rules recall the example from the Introduction
\begin{equation}\label{example:minitest}
{\scope{\mathit{exam}}}\scope{\mathit{minitest}}
\scope{\mathit{alice}}\receive{\mathit{alice}}\role\msg{x}.
\receive{x}\role\msg{\mathtt{Task}}%. \mathit{DoTask}
\end{equation}
and the type $\{\mathit{alice}\}( \{\mathit{exam},\mathit{minitest}\} ( \emptyset ))$ assigned to channel name ${\mathit{alice}}$. 
Following typing rule \rulename{(t-in)}, we see that name $\NX$ is typed with $ \{\mathit{exam},\mathit{minitest}\} ( \emptyset )$, i.e., the carried type of ${\mathit{alice}}$. Now, knowing that name $\NX$ can be instantiated by $\mathit{exam}$ or $\mathit{minitest}$ in order to apply rule \rulename{(t-in)}  considering contextual authorizations we must check whether the authorizations to use both names are provided (contained in $\rho$). Furthermore, consider that the process shown in~(\ref{example:minitest}) is composed in parallel with another process willing to send a name along $\mathit{alice}$, specifically
\begin{equation}\label{example:sending.exam}
\scope{\mathit{alice}}\send{\mathit{alice}}\role\msg{exam}. \inact \parop\scope{\mathit{exam}}\scope{minitest}
\scope{\mathit{alice}}\receive{\mathit{alice}}\role\msg{x}.
\receive{x}\role\msg{\mathtt{Task}}. %\mathit{DoTask}.
\end{equation}
Considering the same typing assumption for $\mathit{alice}$ (with type $\{\mathit{alice}\}( \{\mathit{exam}, \mathit{minitest}\} ( \emptyset ))$), and assuming that name $\mathit{exam}$ has type $ \{\mathit{exam}\} ( \emptyset )$, by rule \rulename{(t-out)} we can conclude that it is safe to send name $\mathit{exam}$ along $\mathit{alice}$, since $\{\mathit{exam}\}$ is contained in $\{\mathit{exam},\mathit{minitest}\}$, the (only) instantiation of $\mathit{exam}$  is contained in the carried type of $\mathit{alice}$.

Now consider that the latter process is placed in the context where name $\mathit{exam}$ is restricted 
\begin{equation}\label{example:restricted.name}
\rest{\mathit{exam}:\NR}(\scope{\mathit{alice}}\send{\mathit{alice}}\role\msg{exam}. \inact \parop\scope{\mathit{exam}}\scope{minitest}
\scope{\mathit{alice}}\receive{\mathit{alice}}\role\msg{x}.
\receive{x}\role\msg{\mathtt{Task}}). %\mathit{DoTask})
\end{equation}

In order to type this process, the assumption for $\mathit{alice}$ considers type $\{\mathit{alice}\}( \{\NR, \mathit{minitest}\} ( \emptyset ))$, representing that in $\mathit{alice}$ a restricted name can be communicated. Hence, the process shown in~(\ref{example:restricted.name}) cannot be composed with others that rely on contextual authorizations for names exchanged in $\mathit{alice}$. This follows from the fact that the multiset of provided authorizations $\rho$ by definition can contain only names and not symbols. To use names received along $\mathit{alice}$ one has to specify an authorization for all possible receptions (e.g. $\scope{\mathit{alice}}\receive{\mathit{alice}}\role\msg{\NX}.\scope\NX\receive\NX\role\msg{\mathtt{Task}}$), or rely on authorization delegation.

Let us also consider process 
\begin{equation}\label{example:rep.in.with.r}
\repreceive{{\mathit{license}}}\role\msg\NX.\rest{\mathit{exam}: \NR}(\scope\NX\send\NX\role\msg{\mathit{exam}}.\inact \parop \scope\NX\scope{\mathit{exam}}\receive\NX\role\msg\NY.\send\NY\role\msg{\mathit{task}}.\inact)
\end{equation}
that models a server that receives a name and afterwards it is capable of both receiving (on the lhs) and sending (a fresh name, on the rhs) along the received name. Our typing analysis excludes this process since it contains a symbol ($\NR$) in a replicated input (see rule \rulename{(t-rep-in)}). To show why names bound inside a replicated input cannot be subject to contextual authorizations, even inside the scope of the restriction (like in~(\ref{example:restricted.name})), consider that the process shown in~(\ref{example:rep.in.with.r}) may evolve by receiving name  $\mathit{alice}$ twice, to  

\begin{equation}\label{example:error.rep}
\begin{array}{c}

\rest{\mathit{exam_1}: \NR}(\scope{\mathit{alice}}\send{\mathit{alice}}\role\msg{\mathit{exam_1}}.\inact \parop \scope{\mathit{exam_1}}\scope{\mathit{alice}}\receive{\mathit{alice}}\role\msg\NY.\send\NY\role\msg{\mathit{task}}.\inact) 
\\
\parop
\rest{\mathit{exam_2}: \NR}(\scope{\mathit{alice}}\send{\mathit{alice}}\role\msg{\mathit{exam_2}}.\inact \parop \scope{\mathit{exam_2}}\scope{\mathit{alice}}\receive{\mathit{alice}}\role\msg\NY.\send\NY\role\msg{\mathit{task}}.\inact). 
\end{array}
\end{equation}
where two copies of the replicated process are active in parallel, and where two different restricted names can be sent on $\mathit{alice}$, hence leading to an error when the contextual authorization does not match the received name (e.g., $\scope{\mathit{exam_2}}\scope{\mathit{alice}}\send{\mathit{exam_1}}\role\msg{\mathit{task}}.\inact$).

In order to address name generation within replicated input we use distinguished symbol $\nu$ that captures the fact that such names are never subject to contextual authorizations. This means that names typed with $\nu$ both within their scope and outside of it cannot be granted contextual authorizations. Hence the process obtained by replacing the $\NR$ annotation by $\nu$ in~(\ref{example:rep.in.with.r}), concretely
\begin{equation}\label{example:rep.in.with.nu}
\repreceive{{\mathit{license}}}\role\msg\NX.\rest{\mathit{exam}: \nu}(\scope\NX\send\NX\role\msg{\mathit{exam}}.\inact \parop \scope\NX\scope{\mathit{exam}}\receive\NX\role\msg\NY.\send\NY\role\msg{\mathit{task}}.\inact)
\end{equation}
is also not typable since a contextual authorization is expected for name $\mathit{exam}$. Notice this configuration leads to an error like the one shown in~(\ref{example:error.rep}). However process
\begin{equation}\label{example:rep.in.ok}
\repreceive{{\mathit{license}}}\role\msg\NX.\rest{\mathit{exam}: \nu}\scope{\mathit{alice}}\send{\mathit{alice}}\role\msg{\mathit{exam}}.\inact
\end{equation}
is typable, hence may be safely composed with contexts that comply with assumption $\mathit{alice}: \{\mathit{alice}\}(\nu(T))$, i.e., do not rely on contextual authorizations for names received on $\mathit{alice}$.

We remark that our approach can be generalized to address contextual authorizations for name 
generation in the context of certain forms of infinite behavior, namely considering recursion 
together with linearity constraints that ensure race freedom (like in the setting of behavioral types~\cite{bettyreport}).

%\input{types}
%\input{checking}
%!TEX root =  auth.tex

\section{Concluding remarks}\label{sec:Conclusions}

In the literature we find a plethora of techniques that address resource usage control, ranging 
from locks that guarantee mutual exclusion in critical code blocks to communication protocols 
(e.g., token ring). Several typing disciplines have been developed to ensure proper resource usage, such
as~\cite{DBLP:journals/corr/abs-1712-08310,DBLP:journals/jlp/GorlaP09,DBLP:conf/esop/SwamyCC10}, 
where capabilities are specified in the type language, not as a first class entity in the model. 
Therefore in such approaches it is not possible to separate resource and capability like we 
do---cf. the ``unauthorized'' intermediaries such as brokers mentioned in the Introduction. We distinguish 
an approach that considers accounting~\cite{DBLP:journals/corr/abs-1712-08310}, in the sense that 
the types specify the number of messages that may be exchanged, which therefore relates to the notion 
of accounting presented here. 

We also find proposals of models that include capabilities as first class entities, addressing usage
of channels and of resources as communication objects, such 
as~\cite{DBLP:journals/scp/BodeiDF17,Giunti,DBLP:journals/lmcs/KobayashiSW06,VivasYoshida}. 
More specifically, constructs for restricting (hiding and filtering) the behaviors allowed on channels~\cite{Giunti,VivasYoshida}, usage specification in a (binding) name scope 
construct~\cite{DBLP:journals/lmcs/KobayashiSW06}, and authorization scopes for resources based 
on given access policies~\cite{DBLP:journals/scp/BodeiDF17}. We distinguish our 
approach from~\cite{Giunti,VivasYoshida} since the proposed constructs are static and not able to 
capture our notion of a floating resource capability. As for~\cite{DBLP:journals/lmcs/KobayashiSW06}, 
the usage specification directly embedded in the model resembles a type and is given in a binding 
scoping operator which contrasts with our non-binding authorization scoping.
Also in~\cite{DBLP:journals/scp/BodeiDF17} detailed usage policies 
are provided associated to the authorization scopes for resources. 
In both~\cite{DBLP:journals/scp/BodeiDF17,DBLP:journals/lmcs/KobayashiSW06} the models seem less 
adequate to capture our notion of floating authorizations as access is granted explicitly and 
controlled via the usage/policy specification, and for instance our notion of confinement does not 
seem to be directly representable.

We believe our approach can be extended by considering some form of usage specifications like
the ones mentioned above~\cite{DBLP:journals/scp/BodeiDF17,DBLP:journals/lmcs/KobayashiSW06}, 
by associating to each authorization scoping more precise capabilities in the form of behavioral 
types~\cite{bettyreport}. It would also be interesting to resort to refinement 
types~\cite{DBLP:conf/pldi/FreemanP91} to carry out our typing analysis, given that our types 
can be seen to some extent as refinements on the domain of names, an investigation we leave
to future work. Perhaps even more relevant would be to convey our principles to the licensing
domain where we have identified related patents~\cite{armstrong2005management,baratti2003license} for the purpose of 
certifying license usage. At this level it would be important to extend our work considering
also non-consumptive authorizations in the sense of authorizations that can be placed back
to their original scope after they have been used.

We have presented a model that addresses floating authorizations, a notion we believe is unexplored
in the existing literature. We based our development on previous work~\cite{clar:eke} by extending the
model in a minimal way so to carry out our investigation, even though the required technical changes 
revealed themselves to be far from straightforward. We left out non-determinism in the form of choice 
since our focus is on the interplay between parallel composition and authorization scope, and we 
believe adding choice to our development can be carried out in expected lines. We remark that 
in~\cite{clar:eke} a certain form of accounting inconsistency when handling authorization delegation was 
already identified, while in this work we believe accounting is handled consistently throughout the model. 
We intend to study the behavioral theory of our model, also for the sake of illuminating
our notion of floating authorizations and their accounting, where for instance an axiomatization of 
the behavioral semantics would surely be informative on the authorization scoping construct.

We also presented a typing analysis that addresses contextual authorizations, which we also believe
is unexplored in the literature in the form we present it here. Our typing rules induce a decidable 
type-checking procedure, since rules are syntax directed, provided as usual that a (carried) type 
annotation is added to name restrictions. However, we have already started working on a 
type-checking procedure nevertheless based on our typing rules but where the focus is on efficiency, 
namely at the level of distributing authorizations provided by the environment to (parallel) subsystems.
This allows for fine-grained information on authorizations actually required by processes, which will 
hopefully lead to identifying principles that may be used for the sake of type inference.
A notion of substitutability naturally arises in our typing analysis and we leave to future work a detailed 
investigation of a subtyping relation that captures such notion, but we may mention that our preliminary 
assessment actually hinted on some non standard features with respect to variance and covariance of 
carried types.

%
% ---- Bibliography ----
%
\bibliographystyle{plain}
\bibliography{bib}

\newpage
\appendix
%!TEX root =  auth.tex
%\setcounter{lemma}{section}
\section{Appendix}

\subsection{Proofs from Section~\ref{sec:Calculus}}
\label{app:semproofs}

\begin{lemma}[Inversion on Labelling]{\label{lemm:free_names}}
Let $\PP\lts{\alpha}\PQ.$
\begin{itemize}
\item[1.] If $\alpha=\scope\NA^i\send\NA\role\msg\NB$ then $\NA,\NB\in\fn\PP.$
\item[2.] If $\alpha=\rest\NB\scope\NA^i\send\NA\role\msg\NB$ then $\NA\in\fn\PP$ and $\NB\in\bn\PP.$
\item[3.] If $\alpha=\scope\NA^i\receive\NA\role\msg\NB$ then $\NA\in\fn\PP.$
\item[4.] If $\alpha=\scope\NA^i\scope\NB^j\sauth\NA\role\msg\NB$ then $\NA,\NB\in\fn\PP.$
\item[5.] If $\alpha=\scope\NA^i\rauth\NA\role\msg\NB$ then $\NA,\NB\in\fn\PP.$
\end{itemize}
\end{lemma}
\begin{proof}
The proof is by induction on the inference of $\PP\lts{\alpha}\PQ.$ We comment just the first and the second assertions.

$\mathit{1.}$ The base case for  $\PP\lts{\send\NA\role\msg\NB}\PQ$ is the rule \rulename{(l-in)}, then $\PP=\send\NA\role\msg\NB.\PQ$ and $\NA,\NB\in\fn\PP.$ We comment only the case when the last applied rule is \rulename{(l-res)}: $\rest\NC\PP'\lts{\send\NA\role\msg\NB}\rest\NC\PQ'$ is derived from $\PP'\lts{\send\NA\role\msg\NB}\PQ'$, where $\NC\notin\{\NA,\NB\}$ and by the induction hypothesis $\NA,\NB\in\fn{\PP'},$ which implies $\NA,\NB\in\fn{\PP'}\setminus\{\NC\}=\fn{\rest\NC\PP'}.$  The base case for $\alpha=\scope\NA\send\NA\role\msg\NB$ is when $\scope\NA\PP'\lts{\scope\NA\send\NA\role\msg\NB}\PQ$ is derived from $\PP'\lts{\send\NA\role\msg\NB}\PQ$ by \rulename{(l-scope-ext)}. Then, by the first part of the proof $\NA,\NB\in\fn{\PP'},$ thus $\NA,\NB\in\fn{\scope\NA\PP'}.$

$\mathit{2.}$ We only comment the base case: $\rest\NB\PP\lts{\rest\NB\scope\NA^i\send\NA\role\msg\NB}\PQ$ is derived from $\PP\lts{\scope\NA^i\send\NA\role\msg\NB}\PQ,$ where $\NA\not=\NB,$ by the rule \rulename{(l-open)}, then from the first assertion of this Lemma we get $\NA,\NB\in\fn\PP.$ From this it directly follows $\NA\in\fn{\rest\NB\PP}$ and $\NB\in\bn{\rest\NB\PP}.$
\end{proof}

\begin{lemma}[LTS Closure Under Structural Congruence] {\label{lemm:lts_struct_congruence}}
If $\PP\equiv\PP'$ and $\PP\lts{\alpha}\PQ,$ then there exists some $\PQ'$ such that $\PP'\lts{\alpha}\PQ'$ and $\PQ\equiv\PQ'.$
\end{lemma}

\begin{proof}
The proof is by induction on the length of the derivation of $\PP\equiv\PP'.$ 
We detail only the case when the last applied rule is \rulename{(sc-res-extr)} $\PP_1\parop\rest\NA\PP_2\equiv\rest\NA(\PP_1\parop\PP_2)$ if $\NA\notin\fn{\PP_1}.$ 
We have two possible transitions for $\rest\NA(\PP_1\parop\PP_2)$, by \rulename{(l-res)} and \rulename{(l-open)}:
\begin{itemize}
\item \rulename{(l-res)}: Assume that $\rest\NA(\PP_1\parop\PP_1)\lts{\alpha} \rest\NA\PR,$ where $\NA\notin\n\alpha$ is derived from $(\PP_1\parop\PP_2)\lts{\alpha} \PR.$ Then, possible transitions for $\PP_1\parop\PP_2$ are:
	\begin{itemize}
	\item \rulename{(l-par)}: $\PP_1\parop\PP_2\lts{\alpha}\PP_1\parop\PP_2',$ where $\bn{\alpha}\cap\fn{\PP_1}=\emptyset$ is derived from $\PP_2\lts{\alpha}\PP_2'.$ Since $\NA\notin\n\alpha,$ by \rulename{(l-res)} we get $\rest\NA\PP_2\lts{\alpha}\rest\NA\PP_2'$ and by \rulename{(l-par)} $\PP_1\parop\rest\NA\PP_2\lts{\alpha}\PP_1\parop\rest\NA\PP_2'.$ If  $\PP_1\parop\PP_2\lts{\alpha}\PP_1'\parop\PP_2,$ where $\bn{\alpha}\cap\fn{\PP_2}=\emptyset$ and $\PP_1\lts{\alpha}\PP_1'$, then, since $\fn{\rest\NA\PP_2}\subseteq\fn{\PP_2},$ by \rulename{(l-par)} we get $\PP_1\parop\rest\NA\PP_2\lts{\alpha}\PP_1'\parop\rest\NA\PP_2.$
	\item \rulename{(l-comm)}: $\PP_1\parop\PP_2\lts{\alpha}\PP_1'\parop\PP_2'$ and $\alpha=\tau_{\omega},$ where $\omega=\scope\NB^{2-i-j}$ is derived from  $\PP_1\lts{\alpha_1}\PP_1'$ and $\PP_2\lts{\alpha_2}\PP_2',$ for $\alpha_1,\alpha_2\in\{\scope\NB^i\send\NB\role\msg\NC, \scope\NB^j\receive\NB\role\msg\NC\}.$ By Lemma~\ref{lemm:free_names} we get $\NB\in\fn{\PP_1,\PP_2}.$ Thus from $\NA\notin\fn{\PP_1}$ we get $\NB\not=\NA.$ Then, we have two cases:
		\begin{itemize}
		\item if $\NC=\NA$ and  $\alpha_2=\scope\NB^i\send\NB\role\msg\NA.$ Then, by \rulename{(l-open)} $\rest\NA\PP_2\lts{\rest\NA\alpha_2}\PP_2'$ and by \rulename{(l-close)} $\PP_1\parop\rest\NA\PP_2\lts{\tau_\omega}\rest\NA(\PP_1'\parop\PP_2');$
		\item if $\NC\not=\NA,$ then by \rulename{(l-res)} $\rest\NA\PP_2\lts{\alpha_2}\rest\NA\PP_2'$ and by \rulename{(l-comm)} $\PP_1\parop\rest\NA\PP_2 \lts{\tau_\omega}\PP_1'\parop\rest\NA\PP_2'.$ 
		\end{itemize}
	\item \rulename{(l-close)}: $\PP_1\parop\PP_2\lts{\alpha}\rest\NC(\PP_1'\parop\PP_2')$ and $\alpha=\tau_{\omega},$ where $\omega=\scope\NB^{2-i-j}$ is derived from  $\PP_1\lts{\alpha_1}\PP_1'$ and $\PP_2\lts{\alpha_2}\PP_2',$ for $\alpha_1,\alpha_2\in\{\rest\NC\scope\NB^i\send\NB\role\msg\NC, \scope\NB^j\receive\NB\role\msg\NC\}.$ By Lemma~\ref{lemm:free_names} we get $\NB\in\fn{\PP_1,\PP_2}$ and $\NC\in\bn{\PP_1,\PP_2}.$ Since $\NA\notin\fn{\PP_1}$ and assuming all bound names are different, we get $\NA\notin\{\NB,\NC\}.$ Thus, by \rulename{(l-res)} $\rest\NA\PP_2\lts{\alpha_2}\rest\NA\PP_2'$ and by \rulename{(l-close)} $\PP_1\parop\rest\NA\PP_2 \lts{\tau_\omega}\rest\NC(\PP_1'\parop\rest\NA\PP_2').$
	\item \rulename{(l-auth)}:  $\PP_1\parop\PP_2\lts{\alpha}\PP_1'\parop\PP_2'$ and $\alpha=\tau_{\omega},$ where $\omega=\scope\NB^{2-i-j}\scope\NC^{1-k}$ is derived from $\PP_1\lts{\alpha_1}\PP_1'$ and $\PP_2\lts{\alpha_2}\PP_2',$ for $\alpha_1,\alpha_2\in\{\scope\NB^i\scope\NC^k\sauth\NB\role\msg\NC, \scope\NB^j\rauth\NB\role\msg\NC\}.$ By Lemma~\ref{lemm:free_names} we get $\NC,\NB\in\fn{\PP_1,\PP_2}.$ Since $\NA\notin\fn{\PP_1}$ we get $\NA\notin\{\NB,\NC\}.$ Thus, by \rulename{(l-res)} $\rest\NA\PP_2\lts{\alpha_2}\rest\NA\PP_2'$ and by \rulename{(l-auth)} $\PP_1\parop\rest\NA\PP_2 \lts{\tau_\omega}\PP_1'\parop\rest\NA\PP_2'.$
	\end{itemize}
\item \rulename{(l-open)}: Assume that $\rest\NA(\PP_1\parop\PP_2)\lts{\rest\NA\scope\NB^i\send\NB\role\msg\NA} \PR,$ is derived from\\
 $\PP_1\parop\PP_2\lts{\scope\NB^i\send\NB\role\msg\NA} \PR,$ where $\NA\not=\NB.$ Since $\NA\notin\fn{\PP_1}$ and by Lemma~\ref{lemm:free_names} we get that $\PP_1\parop\PP_2$ could be only derived using \rulename{(l-par)} $\PP_1\parop\PP_2\lts{\scope\NB^i\send\NB\role\msg\NA} \PP_1\parop \PP_2'$ from $\PP_2\lts{\scope\NB^i\send\NB\role\msg\NA}\PP_2'$. By \rulename{(l-open)} $\rest\NA\PP_2\lts{\rest\NA\scope\NB^i\send\NB\role\msg\NA}\PP_2'$ and by \rulename{(l-par)} we get $\PP_1\parop\rest\NA\PP_2\lts{\rest\NA\scope\NB^i\send\NB\role\msg\NA}\PP_1\parop\PP_2'.$
\end{itemize}
\end{proof}

\begin{lemma}[Inversion on LTS]{\label{lemm:inv_on_lts}}
Let $\PP$ and $\PQ$ be processes.
\begin{itemize}
\item[1.] If $\PP\lts{\send\NA\role\msg\NB}\PQ$ then\\
 $\PP\equiv\rest{\tilde {\ND}}\context[\send\NA\role\msg\NB.\PP']$ and $\PQ\equiv\rest{\tilde {\ND}}\context[\scope{\NA}\PP']$ and $\operator(\context[\cdot]; \NA)$ is undefined;
\item[2.] If $\PP\lts{\scope{\NA}\send\NA\role\msg\NB}\PQ$ then\\
 $\PP\equiv\rest{\tilde {\ND}}\context[\send\NA\role\msg\NB.\PP']$ and $\PQ\equiv\rest{\tilde {\ND}}\context^-[\scope{\NA}\PP'],$ for $\context^-[\cdot]=\operator(\context[\cdot]; \NA);$
\item[3.] If $\PP\lts{\rest\NB\send\NA\role\msg\NB}\PQ$ then \\
$\PP\equiv\rest{\tilde {\ND}}\rest\NB\context[\send\NA\role\msg\NB.\PP']$ and 
$\PQ\equiv\rest{\tilde {\ND}}\context[\scope{\NA}\PP']$ and $\operator(\context[\cdot]; \NA)$ is undefined;
\item[4.] If $\PP\lts{\rest\NB\scope{\NA}\send\NA\role\msg\NB}\PQ$ then \\
$\PP\equiv\rest{\tilde {\ND}}\rest\NB\context[\send\NA\role\msg\NB.\PP']$ and 
$\PQ\equiv\rest{\tilde {\ND}}\context^-[\scope{\NA}\PP'],$ for $\context^-[\cdot]=\operator(\context[\cdot]; \NA);$
\item[5.] If $\PP\lts{\receive\NA\role\msg\NB}\PQ$ then 
$\PP\equiv\rest{\tilde {\ND}}\context[\receive\NA\role\msg\NX.\PP']$ and 
$\PQ\equiv\rest{\tilde {\ND}}\context[\scope{\NA}\PP'\subst{\NB}{\NX}]$ and $\operator(\context[\cdot]; \NA)$ is undefined; 
\item[6.] If $\PP\lts{\scope{\NA}\receive\NA\role\msg\NB}\PQ$ then 
$\PP\equiv\rest{\tilde {\ND}}\context[\receive\NA\role\msg\NX.\PP']$ and 
$\PQ\equiv\rest{\tilde {\ND}}\context^-[\scope{\NA}\PP'\subst{\NB}{\NX}],$ for $\context^-[\cdot]=\operator(\context[\cdot]; \NA),$ 
\item[7.] If $\PP\lts{\scope{\NA}\sauth\NA\role\msg\NB}\PQ$ then 
$\PP\equiv\rest{\tilde {\ND}}\context[\sauth\NA\role\msg\NB.\PP']$ and  
$\PQ\equiv\rest{\tilde {\ND}}\context^-[\scope\NA\PP'],$ for $\context^-[\cdot]=\operator(\context[\cdot]; \NA)$ and
 $\operator(\context[\cdot]; \NA,\NB)$ is undefined;
\item[8.] If $\PP\lts{\scope{\NB}\sauth\NA\role\msg\NB}\PQ$ then 
$\PP\equiv\rest{\tilde {\ND}}\context[\sauth\NA\role\msg\NB.\PP']$ and  
$\PQ\equiv\rest{\tilde {\ND}}\context^-[\scope\NA\PP'],$ for $\context^-[\cdot]=\operator(\context[\cdot]; \NB)$ and $\operator(\context[\cdot]; \NA,\NB)$ is undefined;
\item[9.]  If $\PP\lts{\scope{\NA}\scope{\NB}\sauth\NA\role\msg\NB}\PQ$ then 
$\PP\equiv\rest{\tilde {\ND}}\context[\sauth\NA\role\msg\NB.\PP']$ and  
$\PQ\equiv\rest{\tilde {\ND}}\context^-[\scope\NA\PP'],$ for $\context^-[\cdot]=\operator(\context[\cdot]; \NA,\NB);$
\item[10.] If $\PP\lts{\rauth\NA\role\msg\NB}\PQ$ then 
$\PP\equiv\rest{\tilde {\ND}}\context[\rauth\NA\role\msg\NB.\PP']$ and 
$\PQ\equiv\rest{\tilde {\ND}}\context[\scope\NA\scope{\NB}\PP']$ and \\
$\operator(\context[\cdot]; \NA)$ is undefined;
\item[11.] If $\PP\lts{\scope{\NA}\rauth\NA\role\msg\NB}\PQ$ then 
$\PP\equiv\rest{\tilde {\ND}}\context[\rauth\NA\role\msg\NB.\PP']$ and 
$\PQ\equiv\rest{\tilde {\ND}}\context^-[\scope\NA\scope{\NB}\PP'],$ for $\context^-[\cdot]=\operator(\context[\cdot]; \NA);$
\item[12.] If $\PP\lts{\tau}\PQ$ then either
	\begin{itemize}
	\item $\PP\equiv \rest{\tilde {\ND}}\context[\send\NA\role\msg\NB.\PP', \receive\NA\role\msg\NX.\PP'']$ and \\
	$\PQ\equiv\rest{\tilde {\ND}}\context^-[\scope\NA\PP', \scope\NA\PP''\subst{\NB}{\NX}],$ for $\context^-[\cdot_1,\cdot_2]=\operator(\context[\cdot_1, \cdot_2 ]; \NA;\NA),$ or 
	\item $\PP\equiv \rest{\tilde {\ND}}\context[\sauth\NA\role\msg\NB.\PP', \rauth\NA\role\msg\NB.\PP'']$ and \\
	$\PQ\equiv \rest{\tilde {\ND}}\context^-[\scope\NA\PP', \scope\NA\scope\NB\PP''],$ for $\context^-[\cdot_1,\cdot_2]=\operator(\context[\cdot_1, \cdot_2 ]; \NA,\NB;\NA);$
	\end{itemize}
\item[13.] If $\PP\lts{\tau_{\scope\NA}}\PQ$ then either
	\begin{itemize}
	\item  $\PP\equiv \rest{\tilde {\ND}}\context[\send\NA\role\msg\NB.\PP', \receive\NA\role\msg\NX.\PP'']$ and $\PQ\equiv\rest{\tilde {\ND}}\context^-[\scope\NA\PP', \scope\NA\PP''\subst{\NB}{\NX}],$ for \\
	$\context^-[\cdot_1,\cdot_2]=\operator(\context[\cdot_1, \cdot_2 ];  \NA; \emptyset),$ or $\context^-[\cdot_1,\cdot_2]=\operator(\context[\cdot_1, \cdot_2 ]; \emptyset; \NA),$ and \\
	$\operator(\context[\cdot_1, \cdot_2 ]; \NA;\NA)$ is undefined, or
	\item $\PP\equiv \rest{\tilde {\ND}}\context[\sauth\NA\role\msg\NB.\PP', \rauth\NA\role\msg\NB.\PP'']$ and $\PQ\equiv \rest{\tilde {\ND}}\context^-[ \scope\NA\PP', \scope\NA\scope\NB\PP''],$ for \\
	 $\context^-[\cdot_1,\cdot_2]=\operator(\context[\cdot_1, \cdot_2 ];\NA,\NB; \emptyset),$ or $\context^-[\cdot_1,\cdot_2]=\operator(\context[\cdot_1, \cdot_2 ]; \NB; \NA),$ and\\
	 $\operator(\context[\cdot_1, \cdot_2 ]; \NA,\NB;\NA)$ is undefined, or
	\item $\PP\equiv \rest{\tilde {\ND}}\context[\sauth\NB\role\msg\NA.\PP', \rauth\NB\role\msg\NA.\PP'']$ and $\PQ\equiv \rest{\tilde {\ND}}\context^-[\scope\NB\PP', \scope\NB\scope\NA\PP''],$ for \\
	$\context^-[\cdot_1,\cdot_2]=\operator(\context[\cdot_1, \cdot_2 ]; \NB;\NB)$  and $\operator(\context[\cdot_1, \cdot_2 ]; \NB,\NA;\NB)$ is undefined.
	\end{itemize}
	\end{itemize}
\end{lemma}
\begin{proof}
The proof is by induction on the inference of $\PP\lts{\alpha}\PQ.$ We comment just the first two assertions.

$\mathit{1.}$ Suppose $\PP\lts{\send\NA\role\msg\NB}\PQ$ and let us show that  $\PP\equiv\rest{\tilde {\ND}}\context[\send\NA\role\msg\NB.\PP']$ and $\PQ\equiv\rest{\tilde {\ND}}\context[\scope\NA\PP']$ and $\operator(\context[\cdot]; \NA)$ is undefined.
 We get the base case by rule \rulename{(l-out)} $\send\NA\role\msg\NB\PP\lts{\send\NA\role\msg\NB}\scope\NA\PP.$ Here, $\send\NA\role\msg\NB\PP=\context[\send\NA\role\msg\NB\PP]$ and $\scope\NA\PP=\context[\scope\NA\PP],$ where $\context[\cdot]=\cdot.$ The operator $\operator(\context[\cdot]; \NA)$ is undefined, since the second parameter of the operator is not an empty list.
For induction steps we have next cases of the last applied rule: \rulename{(l-res)}, \rulename{(l-scope)} and \rulename{(l-par)}. If the last applied rule is \rulename{(l-res)}, we immediately get the result from the induction hypothesis. If the last applied rule is \rulename{(l-scope)}, we get $\PP\lts{\send\NA\role\msg\NB}\PQ$ from the $\scope\NC\PP\lts{\send\NA\role\msg\NB}\scope\NC\PQ$, where $\NC\not=\NA.$ By induction hypothesis $\PP\equiv\rest{\tilde {\ND}}\context[\send\NA\role\msg\NB.\PP']$ and $\PQ\equiv\rest{\tilde {\ND}}\context[\scope\NA\PP']$ and $\operator(\context[\cdot]; \NA)$ is undefined. Considering all free and bound names are different, by \rulename{(sc-scope-auth)} we get $\scope\NC\PP\equiv\rest{\tilde {\ND}}\scope\NC\context[\send\NA\role\msg\NB.\PP']$ and $\scope\NC\PQ\equiv\rest{\tilde {\ND}}\scope\NC\context[\scope\NA\PP'].$ For $\context'[\cdot]=\scope\NC\context[\cdot]$ we get that  $\scope\NC\PP\equiv\rest{\tilde {\ND}}\context'[\send\NA\role\msg\NB.\PP']$ and $\scope\NC\PQ\equiv\rest{\tilde {\ND}}\context'[\scope\NA\PP']$ and $\operator(\context'[\cdot]; \NA)$ is undefined since $\NC\not=\NA.$ Case  \rulename{(l-par)} we get by similar reasoning.

$\mathit{2.}$ Consider now $\PP\lts{\scope\NA\send\NA\role\msg\NB}\PQ.$
We get the base case by rule \rulename{(l-scope-int)}: $\scope\NA\PP\lts{\scope\NA\send\NA\role\msg\NB}\PQ$ is derived from $\PP\lts{\send\NA\role\msg\NB}\PQ.$ 
By the first part of the proof we get $\PP\equiv\rest{\tilde {\ND}}\context[\send\NA\role\msg\NB.\PP']$ and $\PQ\equiv\rest{\tilde {\ND}}\context[\scope\NA\PP']$ and $\operator(\context[\cdot]; \NA)$ is undefined. 
By Lemma~\ref{lemm:free_names} we get $\NA\in \fn\PP$, thus $\NA\not\in\tilde {\ND}.$ 
Then, by \rulename{(sc-scope-auth)} we get $\scope\NA\PP\equiv\rest{\tilde {\ND}}\scope\NA\context[\send\NA\role\msg\NB.\PP'].$ 
For $\context'[\cdot]=\scope\NA\context[\cdot],$ we get  $\operator(\context'[\cdot]; \NA)=\context[\cdot],$ which implies $\PQ\equiv\rest{\tilde {\ND}}\context[\scope{\NA}\PP'].$ Again, for the inductions steps we have same possibilities for the last applied rule : \rulename{(l-res)}, \rulename{(l-scope)} and \rulename{(l-par)}, all three cases are similar as in the first part of the proof.
\end{proof}

\begin{proposition}[Inversion on Drift]{\label{prop:Cases_for_contexts}}
\begin{itemize}
\item[1.] If $\optop{\context[\cdot_1, \cdot_2]}{\NA}{\NA}$ is defined and $\context[\cdot_1, \cdot_2]= \context'[\context_1[\cdot] \parop \context_2[\cdot]]$ then:
\begin{itemize}
\item {\bf(case $\scope \NA$ in $\context_1$ and $\context_2$):} $\context_1[\cdot]=\context_1'[\scope\NA \context_1''[\cdot]]$ and $\context_2[\cdot]=\context_2'[\scope\NA \context_2''[\cdot]],$ where $\operator(\context_1''[\cdot]; \NA)$ and $\operator(\context_2''[\cdot]; \NA)$ are undefined, or
\item {\bf(case $\scope \NA$ in $\context_1$ and not in $\context_2$):} $\context_1[\cdot]=\context_1'[\scope\NA \context_1''[\cdot]]$ and $\context'[\cdot]=\context_3'[\scope \NA\context_3''[\cdot]],$ where $\operator(\context_1''[\cdot]; \NA),$  $\operator(\context_2[\cdot]; \NA)$ and $\operator(\context_3''[\cdot]; \NA)$ are undefined, or
\item {\bf( case $\scope \NA$ not in $\context_1$ and it is in $\context_2$):} $\context_2[\cdot]=\context_2'[\scope\NA \context_2''[\cdot]]$ and $\context'[\cdot]=\context_3'[\scope \NA\context_3''[\cdot]],$ where $\operator(\context_2''[\cdot]; \NA),$  $\operator(\context_1[\cdot]; \NA)$ and $\operator(\context_3''[\cdot]; \NA)$ are undefined, or
\item {\bf(case $\scope \NA$ not in $\context_1$ and not in $\context_2$):} $\context'[\cdot]=\context_3[\scope\NA \context_3'[\cdot]]$ and $\context_3'[\cdot]=\context_4[\scope\NA \context_4'[\cdot]],$ where $\operator(\context_1[\cdot]; \NA),$ $\operator(\context_2[\cdot]; \NA),$ $\operator(\context_4[\cdot]; \NA, \NA)$ and \\
 $\operator(\context_4'[\cdot]; \NA)$ are undefined.
\end{itemize}

\item[2.] If $\optop{\context[\cdot_1, \cdot_2]}{\NA,\NB}{\NA}$ is defined and $\context[\cdot_1, \cdot_2]= \context'[\context_1[\cdot] \parop \context_2[\cdot]]$ then:
\begin{itemize}
\item {\bf(case $\scope \NA, \scope\NB$ in $\context_1$ and $\scope\NA$ in $\context_2$):} $\context_1[\cdot]=\context_1^{1}[\scope\NA \context_1^{2}[\cdot]]$ and $\context_1^{2}[\cdot]=\context_1^{3}[\scope\NB \context_1^{4}[\cdot]],$ and $\context_2[\cdot]=\context_2'[\scope\NA \context_2''[\cdot]],$ where $\operator(\context_1^2[\cdot]; \NA),$ $\operator(\context_1^4[\cdot]; \NB)$ and $\operator(\context_2''[\cdot]; \NA)$ are undefined, or
\item {\bf(case $\scope \NA, \scope\NB$ in $\context_1$ and $\scope\NA$ in $\context_2$):} $\context_1[\cdot]=\context_1^{1}[\scope\NB \context_1^{2}[\cdot]]$ and $\context_1^{2}[\cdot]=\context_1^{3}[\scope\NA \context_1^{4}[\cdot]],$ and $\context_2[\cdot]=\context_2'[\scope\NA \context_2''[\cdot]],$ where $\operator(\context_1^2[\cdot]; \NB),$ $\operator(\context_1^4[\cdot]; \NA)$ and $\operator(\context_2''[\cdot]; \NA)$ are undefined, or
\item {\bf(case $\scope \NA, \scope\NB$ in $\context_1$ and $\scope\NA$ not in $\context_2$):} $\context_1[\cdot]=\context_1^{1}[\scope\NA \context_1^{2}[\cdot]]$ and $\context_1^{2}[\cdot]=\context_1^{3}[\scope\NB \context_1^{4}[\cdot]],$ and $\context'[\cdot]=\context_3'[\scope \NA\context_3''[\cdot]],$ where $\operator(\context_1^2[\cdot]; \NA),$ $\operator(\context_1^4[\cdot]; \NB),$\\
  $\operator(\context_2[\cdot]; \NA)$ and $\operator(\context_3''[\cdot]; \NA)$ are undefined, or
\item {\bf(case $\scope \NA, \scope\NB$ in $\context_1$ and $\scope\NA$ not in $\context_2$):} $\context_1[\cdot]=\context_1^{1}[\scope\NB \context_1^{2}[\cdot]]$ and $\context_1^{2}[\cdot]=\context_1^{3}[\scope\NA \context_1^{4}[\cdot]]$ and $\context'[\cdot]=\context_3'[\scope \NA\context_3''[\cdot]],$ where $\operator(\context_1^2[\cdot]; \NB),$ $\operator(\context_1^4[\cdot]; \NA),$\\
  $\operator(\context_2[\cdot]; \NA)$ and $\operator(\context_3''[\cdot]; \NA)$ are undefined, or
\item the rest of 10 cases are analog.
\end{itemize}
\end{itemize}
\end{proposition}

\begin{lemma}[Harmony - LTS]{\label{lemm:reduction_tau}}
If $\PP\red\PQ$ then there is $\PQ'$ such that $\PQ\equiv\PQ'$ and $\PP\lts{\tau}\PQ'.$ 
\end{lemma}
\begin{proof}
The proof is by induction on the derivation $\PP\red\PQ.$ We obtain two base cases:
\begin{itemize}
\item \rulename{(r-comm)}:\\
 $\context [\send\NA\role\msg\NB.\PP',\receive\NA\role\msg\NX.\PQ']\red\context^- [\scope{\NA}\PP', \scope{\NA}\PQ'\subst{\NB}{\NX} ],$ for $\context^-[\cdot_1, \cdot_2]=\operator(\context[\cdot_1, \cdot_2]; \NA;\NA).$ By \rulename{(l-out)} and \rulename{(l-in)} we get $\send\NA\role\msg\NB.\PP'\lts{\send\NA\role\msg\NB}\scope\NA\PP'$ and $\receive\NA\role\msg\NX.\PQ'\lts{\receive\NA\role\msg\NB}\scope{\NA}\PQ'\subst{\NB}{\NX}.$ 
 %We use that $\context[\cdot,\cdot]=\context'[\context_1[\cdot]\parop\context_2[\cdot]],$ for some contexts $\context'[\cdot], \context_1[\cdot]$ and $\context_2[\cdot].$ We proceed the proof by induction on the structure of the context $\context[\cdot,\cdot]$ starting with $\context_1[\cdot]$ and $\context_2[\cdot].$
%If $\context_1[\cdot]=\context_1'[\scope\NA\context_1''[\cdot]]$, and $\context_1''[\send\NA\role\msg\NB.\PP']\lts{\send\NA\role\msg\NB}\context_1''[\scope\NA\PP'].$ Then by \rulename{(l-scope1)} we get  $\scope\NA\context_1''[\send\NA\role\msg\NB.\PP']\lts{\scope\NA\send\NA\role\msg\NB}\context_1''[\scope\NA\PP'].$ Since process $\send\NA\role\msg\NB.\PP'$ is not under the scope of $\scope\NA$ in the context $\context_1''[\cdot]$ we have $\operator(\context_1''[\cdot];\NA)$ is undefined and $\operator(\scope\NA\context_1''[\cdot];\NA)=\context_1''[\cdot].$ 
%If $\context_1[\cdot]=\context_1'[\scope\NA\context_1''[\cdot]]$, and $\context_1''[\send\NA\role\msg\NB.\PP']\lts{\scope\NA\send\NA\role\msg\NB}\context_1'''[\scope\NA\PP'].$ Then by \rulename{(l-scope2)} we get  $\scope\NA\context_1''[\send\NA\role\msg\NB.\PP']\lts{\scope\NA\send\NA\role\msg\NB}\scope\NA\context_1'''[\scope\NA\PP'].$ But now since process $\send\NA\role\msg\NB.\PP'$ is under the scope of $\scope\NA$ in the context $\context_1''[\cdot]$ we have $\operator(\context_1''[\cdot];\NA)=\context_1'''[\cdot]$ and $\operator(\scope\NA\context_1''[\cdot];\NA)=\scope\NA\context_1'''[\cdot].$ 
By Proposition~\ref{prop:Cases_for_contexts} we distinguish four cases for the structure of the context $\context[\cdot_1, \cdot_2]=\context'[\context_1[\cdot] \parop \context_2[\cdot]].$ We comment only the case when $\context_1[\cdot]=\context_1'[\scope\NA \context_1''[\cdot]]$ and $\context_2[\cdot]=\context_2'[\scope\NA \context_2''[\cdot]],$ where $\operator(\context_1''[\cdot]; \NA)$ and $\operator(\context_2''[\cdot]; \NA)$ are undefined, thus in contexts $\context_1''[\cdot]$ and $\context_2''[\cdot]$ the holes are not in the scope of authorizations $\scope\NA.$ By induction on contexts $\context_1''[\cdot]$ and $\context_2''[\cdot]$ using rules \rulename{(l-par)} and \rulename{(l-scope)} we get $\context_1''[\send\NA\role\msg\NB.\PP']\lts{\send\NA\role\msg\NB}\context_1''[\scope\NA\PP']$ and $\context_2''[\receive\NA\role\msg\NX.\PQ']\lts{\receive\NA\role\msg\NB}\context_2''[\scope{\NA}\PQ'\subst{\NB}{\NX}].$ By \rulename{(l-scope-ext)} we get $\scope\NA\context_1''[\send\NA\role\msg\NB.\PP']\lts{\scope\NA\send\NA\role\msg\NB}\context_1''[\scope\NA\PP']$ and $\scope\NA\context_2''[\receive\NA\role\msg\NX.\PQ']\lts{\scope\NA\receive\NA\role\msg\NB}\context_2''[\scope{\NA}\PQ'\subst{\NB}{\NX}].$ Proceeding by induction on contexts $\context_1'[\cdot]$ and $\context_2'[\cdot]$ using rules \rulename{(l-par)} and \rulename{(l-scope)} we get $\context_1'[\scope\NA\context_1''[\send\NA\role\msg\NB.\PP']]\lts{\scope\NA\send\NA\role\msg\NB}\context_1'[\context_1''[\scope\NA\PP']]$ and $\context_2'[\scope\NA\context_2''[\receive\NA\role\msg\NX.\PQ']]\lts{\scope\NA\receive\NA\role\msg\NB}\context_2'[\context_2''[\scope{\NA}\PQ'\subst{\NB}{\NX}]].$ By \rulename{(l-comm)} 
\[
\context_1'[\scope\NA\context_1''[\send\NA\role\msg\NB.\PP']] \parop \context_2'[\scope\NA\context_2''[\receive\NA\role\msg\NX.\PQ']]\lts{\tau}\context_1'[\context_1''[\scope\NA\PP']] \parop\context_2'[\context_2''[\scope{\NA}\PQ'\subst{\NB}{\NX}]].
\]
By induction on the structure of context $\context'[\cdot]$ and using rules \rulename{(l-par)} and \rulename{(l-scope)} again, we get 
\[
\context[\send\NA\role\msg\NB.\PP', \receive\NA\role\msg\NX.\PQ']\lts{\tau}
\context'[\context_1'[\context_1''[\scope\NA\PP']] \parop\context_2'[\context_2''[\scope{\NA}\PQ'\subst{\NB}{\NX}]]].
\]
Now we just need to note that $\optop{\context[\cdot_1, \cdot_2]}{\NA}{\NA}=\context'[\context_1'[ \context_1''[\cdot]]\parop\context_2'[ \context_2''[\cdot]].$ 

\item \rulename{(r-auth)} apply similar reasoning.
\end{itemize}

For the induction step we have two possible last applied rules:

\rulename{(r-newc)}: Assume that $\rest\NA\PP\red\rest\NA\PQ$ is derived from $\PP\red\PQ.$ By induction hypothesis $\PP\lts{\tau}\PQ'$ where $\PQ\equiv\PQ'$. By \rulename{(l-res)} we get $\rest\NA\PP\lts{\tau}\rest\NA\PQ'$. By contextually of $\equiv,$ from $\PQ\equiv \PQ'$ we get $\rest\NA\PQ\equiv \rest\NA\PQ'.$

\rulename{(r-stru)}: Assume that $\PP\red\PQ$ is derived from $\PP\equiv\PP'\red\PQ'\equiv\PQ.$ From the induction hypothesis $\PP'\lts{\tau}\PQ'',$ where $\PQ'\equiv\PQ''.$  By Lemma \ref{lemm:lts_struct_congruence} we get $\PP\lts{\tau}\PQ''',$ where $\PQ'''\equiv\PQ''.$ Since the structural congruence is equivalence relation we get $\PQ'''\equiv\PQ''\equiv \PQ'\equiv \PQ.$
\end{proof}

\begin{lemma}[Harmony - Reduction]{\label{lemm:tau_reduction}}
If $\PP\lts{\tau}\PQ$ then $\PP\red\PQ.$
\end{lemma}
\begin{proof}
The proof is by induction on the derivation $\PP\lts{\tau}\PQ$. The base cases are obtained by:
\begin{itemize}
\item \rulename{(l-comm)}: If $\PP_1\parop\PQ_1\lts{\tau}\PP_2\parop\PQ_2$ is derived from $\PP_1\lts{\alpha}\PP_2$ and $\PQ_1\lts{\dual\alpha}\PQ_2,$ where $\alpha,\dual\alpha\in\{\scope\NA\send\NA\role\msg\NB, \scope\NA\receive\NA\role\msg\NB\}.$ By Lemma~\ref{lemm:inv_on_lts}, one possibility is, up to the structural congruence 
\[
\PP_1,\PQ_1\in\{\rest{\tilde {\ND_1}}\context_1[\send\NA\role\msg\NB.\PP_1'], \rest{\tilde {\ND_2}}\context_2[\receive\NA\role\msg\NB.\PQ_1'] \}
\]
and
\[
\PP_2,\PQ_2\in\{\rest{\tilde {\ND}}\context_1^-[\scope{\NA}\PP_1'], \rest{\tilde {\ND_2}}\context_2^-[\scope{\NA}\PQ_1'\subst{\NB}{\NX}]\}.
\]
where $\context_1^-=\operator(\context_1[\cdot]; \NA)$ and $\context_2^-=\operator(\context_2[\cdot]; \NA).$

By \rulename{(sc-res-extr)} we have\\
 $\rest{\tilde {\ND_1}}\context_1[\send\NA\role\msg\NB.\PP_1']\parop \rest{\tilde {\ND_2}}\context_2[\receive\NA\role\msg\NB.\PQ_1'] \equiv \rest{\tilde {\ND_1},\tilde {\ND_2}}(\context[\send\NA\role\msg\NB.\PP_1', \receive\NA\role\msg\NB.\PQ_1'])$ and\\
 $\rest{\tilde {\ND_1}}\context_1^-[\scope{\NA}\PP_1']\parop \rest{\tilde {\ND_2}}\context_2^-[\scope{\NA}\PQ_1'\subst{\NB}{\NX}]$
 $\equiv \rest{\tilde {\ND_1},\tilde {\ND_2}}\context^-[\scope{\NA}\PP_1',\scope{\NA}\PQ_1'\subst{\NB}{\NX}],$\\
 where $\context[\cdot_1, \cdot_2]=\context_1[\cdot_1]\parop\context_2[\cdot_2]$  and $\context^-[\cdot_1, \cdot_2]=\operator(\context'[\cdot_1, \cdot_2];\NA;\NA).$\\
  By \rulename{(r-comm)} and \rulename{(r-newc)} we get \\
 $$\rest{\tilde {\ND_1},\tilde {\ND_2}}\context[\send\NA\role\msg\NB.\PP_1', \receive\NA\role\msg\NB.\PQ_1']\red \rest{\tilde {\ND_1},\tilde {\ND_2}}\context^-[\scope{\NA}\PP_1',\scope{\NA}, \PQ_1'\subst{\NB}{\NX}].$$
 
 The proof is analogous if  
 \[
\PP_1,\PQ_1\in\{\rest{\tilde {\ND_1}}\context_1[\send\NA\role\msg\NB.\PP_1'], \rest{\tilde {\ND_2}}\context_2[\repreceive\NA\role\msg\NB.\PQ_1'] \}
\]
and
\[
\PP_2,\PQ_2\in\{\rest{\tilde {\ND_1}}\context_1^-[\scope{\NA}\PP_1'], \rest{\tilde {\ND_2}}\context_2^-[\scope{\NA}\PQ_1'\subst{\NB}{\NX} \parop \repreceive\NA\role\msg\NB.\PQ_1']\},
\]
where $\context_1^-[\cdot]=\operator(\context_1[\cdot]; \NA)$ and $\context_2^-[\cdot]=\operator(\context_2[\cdot]; \emptyset)=\context_2[\cdot].$

\item \rulename{(l-close)} and  \rulename{(l-auth)} we get by the similar reasoning, where for the second we apply \rulename{(r-auth)} instead of rule \rulename{(r-comm)}.
\end{itemize}
For the induction step,  the last applied rule is one of the following:
\begin{itemize}
\item \rulename{(l-res)}: $\rest\NA\PP\lts{\tau}\rest\NA\PQ$ is derived from $\PP\lts{\tau}\PQ.$ By induction hypothesis $\PP\red\PQ$ and by \rulename{(r-newc)} $\rest\NA\PP\red\rest\NA\PQ$ 
%We get the proof by induction hypothesis and contextually of $\equiv.$ 
\item \rulename{(l-scope-ext)}: $\scope\NA\PP\lts{\tau}\PQ$ is derived from $\PP\lts{\tau_{\scope\NA}}\PQ.$ By Lemma~\ref{lemm:inv_on_lts}, we have three cases, we detail only the case when $\PP\equiv \rest{\tilde {\ND}}\context[\send\NA\role\msg\NB.\PP', \receive\NA\role\msg\NX.\PP'']$ and
 $\PQ\equiv\rest{\tilde {\ND}}\context^-[\scope\NA\PP', \scope\NA\PP''\subst{\NB}{\NX}],$ where $\context^-[\cdot_1, \cdot_2]=\operator(\context[\cdot_1, \cdot_2]; \NA;\emptyset),$ or $\context^-[\cdot_1; \cdot_2]=\operator(\context[\cdot_1, \cdot_2]; \emptyset; \NA),$ and
  $\operator(\context[\cdot_1, \cdot_2]; \NA; \NA)$ is undefined. By\\
   \rulename{(sc-scope-auth)} we get $\scope\NA\PP\equiv \rest{\tilde {\ND}}\scope\NA\context[\send\NA\role\msg\NB.\PP', \receive\NA\role\msg\NX.\PP''].$ \\
 Then $\operator(\scope\NA\context[\cdot_1, \cdot_2]; \NA,\NA)=\context^-[\cdot_1, \cdot_2].$ Thus, by \rulename{(r-comm)}, \rulename{(r-newc)} and \rulename{(r-stru)} we get the proof.
\item \rulename{(l-scope)}:  $\scope\NC\PP\lts{\tau}\scope\NC\PQ$ is derived from $\PP\lts{\tau}\PQ,$ or $\PP\parop\PR\lts{\tau}\PQ\parop\PR$ is derived from $\PP\lts{\tau}\PQ.$ By induction hypothesis $\PP\red\PQ.$ Then 
 $\PP\equiv \rest{\tilde {\ND}}\context[\send\NA\role\msg\NB.\PP_1, \receive\NA\role\msg\NB.\PP_2]$ and 
  $\PQ\equiv \rest{\tilde {\ND}}\context^-[\scope\NA\PP_1, \scope\NA\PP_2\subst{\NB}{\NA}]$, where $\context^-[\cdot_1, \cdot_2]=\operator(\context[\cdot_1, \cdot_2]; \NA; \NA)$ or  
  $\PP\equiv \rest{\tilde {\ND}}\context[\sauth\NA\role\msg\NB.\PP_1, \rauth\NA\role\msg\NB.\PP_2]$  and \\
  $\PQ\equiv \rest{\tilde {\ND}}\context^-[\scope\NA\PP_1, \scope\NA\scope\NB\PP_2]$, where $\context^-[\cdot_1, \cdot_2]=\operator(\context[\cdot_1, \cdot_2]; \NA,\NB; \NA)$. Thus, by \rulename{(sc-scope-auth)} we get $\scope\NC\PP\equiv \rest{\tilde {\ND}}\scope\NC\context[\send\NA\role\msg\NB.\PP_1, \receive\NA\role\msg\NB.\PP_2]$ or \\
  $\scope\NC\PP\equiv \rest{\tilde {\ND}}\scope\NC\context[\sauth\NA\role\msg\NB.\PP_1, \rauth\NA\role\msg\NB.\PP_2]$ where in the former case we get the proof by \rulename{(r-comm)}, \rulename{(r-newc)} and \rulename{(r-stru)}, and in the latter case we get the proof by \rulename{(r-auth)}, \rulename{(r-newc)} and \rulename{(r-stru)}.

\item \rulename{(l-par)}: by similar reasoning as for \rulename{(l-scope)}.
\end{itemize}
\end{proof}

\begin{theorema}{\ref{cor:red=tau}~(Harmony)}

 $\PP\red\PQ$ if and only if $\PP\lts{\tau}\equiv\PQ.$
 
\end{theorema}
\begin{proof}
The proof follows directly from Lemma \ref{lemm:reduction_tau} and Lemma \ref{lemm:tau_reduction}.
\end{proof}

%%%%%%%%%%%%%%%%%%%%%%%%%%%%%%%%%%%%%%%%%%%%%%%%%%%%%%%%%%%%%%%%%%%%%%%%%%

\subsection{Proofs from Section~\ref{sec:Types}}
\label{app:typesproofs}
\begin{proposition}[Preservation of Well-formedness]{\label{lemm:well-formed}}

If $\PP$ is well-formed and $\PP\equiv\PQ$ or $\PP\red\PQ$ then $\PQ$ is also well-formed and $\sym{\PP}=\sym{\PQ}$.
\end{proposition}

\begin{proof}
The proof is by induction on the last applied structural congruence or reduction rule. We detail only the case when the last applied reduction rule is \rulename{(r-newc)}. Let $\PP=\rest{\NA: \NR}\PP',$ $\PQ=\rest{\NA: \NR}\PQ',$ and $\PP\red\PQ$ be derived from $\PP'\red\PQ'.$ Since $\PP$ is well-formed $\NR\notin\sym{\PP'}$ and $\PP'$ is well-formed. By induction hypothesis we get $\PQ'$ is well-formed and $\sym{\PP'}=\sym{\PQ'}.$ Thus, $\NR\notin\sym{\PQ'},$ $\PQ$ is well-formed and $\sym{\PP}=\sym{\PQ}.$
\end{proof}

\begin{lemma}[Inversion on Typing]{\label{lemm:Inv.Lemma}}
\begin{itemize}
\item[1.] If $ \Delta\vdash_\rho \rest{\NA: \NR}\PP$ then  $ \Delta', \NA:\{\NA\}(T)\vdash_\rho \PP ,$ where $\Delta=\Delta'\subst{\NR}{\NA}$ and $\NR\notin \sym{\PP}$ and $\NA\notin\rho, \mathit{names} (T).$ 
\item[2.] If $ \Delta\vdash_\rho \rest{\NA: \nub}\PP$ then  $ \Delta, \NA:\nub(T)\vdash_\rho \PP ,$ where  $\NA\notin\rho, \mathit{names} (T, \Delta).$ 
\item[3.] If $ \Delta\vdash_\rho \scope\NA\PP$ then $ \Delta\vdash_{\rho\uplus \{\NA\}} \PP.$ 
\item[4.] If $ \Delta\vdash_\rho \send\NA\role\msg\NB.\PP$ then $ \Delta\vdash_{\rho} \PP,$ where $\Delta(\NA)=\ttype(\ttype'(T)),$ $\Delta(\NB)=\ttype''(T),$ 	 $\ttype''\subseteq \ttype'$ and if $\NA\notin\rho$ then $\omega\subseteq\rho.$
\item[5.] If $ \Delta\vdash_\rho \receive\NA\role\msg\NX.\PP$ then $ \Delta, \NX:T \vdash_{\rho} \PP,$ where $\Delta(\NA)=\omega(T),$ $\NX\notin\rho, \mathit{names}(\Delta)$ and if $\NA\notin\rho$ then $\omega\subseteq\rho.$
\item[6.] If $ \Delta\vdash_\rho \repreceive\NA\role\msg\NX.\PP$ then $ \Delta, \NX: T \vdash_{\{\NA\}} \PP$ where $\sym{\PP}=\emptyset$ and $\Delta(\NA)=\omega(T)$ and $\NX \notin \rho, \mathit{names}(\Delta).$
\item[7.] If $\Delta\vdash_\rho \sauth\NA\role\msg\NB.\PP$ then $\Delta\vdash_{\rho'} \PP,$ $\Delta(\NA)=\omega(T),$ where $\rho=\rho'\uplus\{\NB\}$ and if $\NA\notin\rho'$ then $\omega \subseteq\rho'.$
\item[8.] If $ \Delta\vdash_\rho \rauth\NA\role\msg\NB.\PP$ then $ \Delta\vdash_{\rho\uplus\{\NB\}} \PP,$ where $\Delta(\NA)=\omega(T)$ and if $\NA\notin\rho$ then $\omega \subseteq\rho.$
\item[9.] If $ \Delta\vdash_{\rho} \PP_1 \parop\PP_2$ then $ \Delta\vdash_{\rho_1}\PP_1$ and $ \Delta\vdash_{\rho_2} \PP_2,$ where $\rho=\rho_1\uplus\rho_2$ and $\sym{\PP_1} \cap \sym{\PP_2}=\emptyset.$
\end{itemize}
\end{lemma}

The following two results (weakening and strengthening properties) are fundamental to prove Subject Congruence, which in turn is crucial to prove Subject Reduction.
We write $\NA \leftrightarrow \NR$ or  $\NA \leftrightarrow \nub$ depending on whether the name $\NA$ is bound in the process, or the process is in a context were the name $\NA$ is bound, with $\rest{\NA :\NR}$ or $\rest{\NA: \nub},$ respectively.

%\comment{H: Please merge the weakening Lemmas (following two), this way we have only one weakening label ;) - you can add a separate statement, hence separate proof, or you can combine the statements and prove separately}
\begin{lemma}[Weakening]{\label{lemm:Weakening_Lemma}}
\begin{enumerate}
\item Let $ \Delta\vdash_\rho \PP$ and $\NA\notin\fn{\PP}\cup\rho.$ Then
\begin{itemize}
\item [1.] if $\NA \leftrightarrow \NR$ and $\NR\notin\sym{\PP}$ and $\Delta'=\Delta\subst{\NA}{\NR}$ then $ \Delta', \NA: \{\NA\}(T) \vdash_\rho \PP;$
\item [2.] if $\NA \leftrightarrow \nub$ then $ \Delta, \NA: \nub(T) \vdash_\rho \PP.$
\end{itemize}
\item  If $ \Delta\vdash_\rho \PP$ then $ \Delta\vdash_{\rho\uplus\rho'} \PP.$
\end{enumerate}
\end{lemma}

\begin{proof}
The proof is by induction on the depth of the derivation $ \Delta\vdash_\rho \PP.$
\\
$\mathit{1.}$ We detail only two cases when the last applied rule is \rulename{(t-out)} or \rulename{(t-rep-in)}. 
\begin{itemize}
\item \rulename{(t-out)}: Let $ \Delta\vdash_\rho \send\NB\role\msg\NC.\PP$ be derived from $ \Delta\vdash_{\rho} \PP,$ where $\Delta(\NB)=\ttype(\ttype'(T')),$ $\Delta(\NC)=\ttype''(T'),$ 	 $\ttype''\subseteq \ttype'.$ and if $\NB\notin\rho$ then $\omega \subseteq\rho.$ 
Then, we have two cases:

-If $\NA \leftrightarrow \NR$ and $\NR\notin \sym{\PP}$ and $\Delta'=\Delta\subst{\NA}{\NR}$, by induction hypothesis $ \Delta', \NA: \{\NA\}(T) \vdash_\rho \PP$. Since $\Delta'(\NB)=(\ttype(\ttype'(T')))\subst{\NA}{\NR},$ $\Delta'(\NC)=(\ttype''(T'))\subst{\NA}{\NR},$ then	 $\ttype''\subst{\NA}{\NR}\subseteq \ttype'\subst{\NA}{\NR}$. Since from $\NB\notin\rho$ it follows $\omega \subseteq\rho,$ then if $\NR\in\omega$ it follows $\NB \in\rho.$ If $\NR \notin \omega$ then $\omega\subst{\NA}{\NR}=\omega.$ Thus, by \rulename{(t-out)} we get $ \Delta', \NA: \{\NA\}(T) \vdash_\rho \send\NB\role\msg\NC.\PP.$

-If $\NA \leftrightarrow \nub$ then by induction hypothesis $ \Delta, \NA: \nub(T) \vdash_\rho \PP.$ Since $\Delta$ does not change we get the result directly by \rulename{(t-out)}.

\item \rulename{(t-rep-in)}: Let $ \Delta\vdash_{\rho} \repreceive\NB\role\msg\NX\PP$ be derived from $ \Delta, \NX: T' \vdash_{\{\NB\}}\PP,$ with  $\Delta(\NB)=\ttype(T'),$ where without loss of generality we can assume that $\NX\not=\NA.$ Again, we have two cases:

-If $\NA \leftrightarrow \NR$ and $\NR\notin \sym{\PP}$ and $\Delta'=\Delta\subst{\NA}{\NR}$, by induction hypothesis $\Delta', \NX:T'\subst{\NA}{\NR}, \NA: \{\NA\}(T) \vdash_{\{\NB\}} \PP.$ Since $\Delta'(\NB)=(\ttype(T'))\subst{\NA}{\NR}$ by \rulename{(t-rep-in)} we get $ \Delta', \NA: \{\NA\}(T) \vdash_{\rho} \repreceive\NB\role\msg\NX\PP.$

-If $\NA \leftrightarrow \nub$ then by induction hypothesis $ \Delta, \NX: T', \NA: \nub(T) \vdash_{\{\NB\}} \PP.$ Again, we get the result directly by \rulename{(t-rep-in)}.
%\htv{maybe distinguish the case when $b=a$ explicitly?}
%\comment{I: I don't think we need this case, I added a condition $\NA\notin\fn{\PP}$ to exclude this case \htv{OK!}}
%\end{proof}

%\comment{H: see comment above}
%\begin{lemmata}{\ref{lemm:Weakening_Lemma}~(Weakening).}
% If $ \Delta\vdash_\rho \PP$ then $ \Delta\vdash_{\rho\uplus\rho'} \PP.$
%\end{lemmata}
%\begin{proof}
%The proof is by induction on the depth of the derivation $ \Delta\vdash_\rho \PP.$ 
\end{itemize}
$\mathit{2.}$ We detail only the case when the last applied rule is \rulename{(t-in)}. Let $ \Delta\vdash_\rho \receive\NA\role\msg\NX.\PP$ be derived from $ \Delta, \NX:T \vdash_{\rho} \PP,$ where $\Delta(\NA)=\omega(T),$ $\NX\notin\rho, \mathit{names}(\Delta)$ and if $\NA\notin\rho$ then $\omega \subseteq\rho.$ By induction hypothesis $ \Delta, \NX: T \vdash_{\rho\uplus\rho'} \PP.$ Without loss of generality we can assume that $\NX$ is new to $\rho',$ i.e. $\NX\notin\rho'.$ Since we can conclude that if  $\NA\notin\rho\uplus\rho'$ then $\omega \subseteq\rho\uplus\rho',$ by \rulename{(t-in)} we get $ \Delta\vdash_{\rho\uplus\rho'} \receive\NA\role\msg\NX.\PP.$
\end{proof}

\begin{lemma}[Strengthening]{\label{lemm:Strength.Lemmma}}
\begin{itemize}
\item [1.]If $ \Delta, \NA: \{\NA\}(T) \vdash_{\rho} \PP$ and $\NA\notin\fn\PP \cup \rho$ and $\NA \leftrightarrow \NR$ and $\NR\notin \sym{\PP}$ then $ \Delta'\vdash_{\rho} \PP,$ where $\Delta'=\Delta\subst{\NR}{\NA}.$
\item [2.] If $ \Delta, \NA: \nub (T) \vdash_{\rho} \PP$ and $\NA\notin\fn\PP \cup \rho$ then $ \Delta\vdash_{\rho} \PP.$ 
\end{itemize}
\end{lemma}

\begin{proof}
The proof is by induction on the depth of the derivation $ \Delta\vdash_\rho \PP.$ We comment only the case when the last applied rule is \rulename{(t-in)}. We have two cases:

-If $\NA \leftrightarrow \NR$ and $\NR\notin \sym{\PP}$  and $\Delta'=\Delta\subst{\NR}{\NA},$ and $ \Delta,\NA: \{\NA\} (T) \vdash_\rho \receive\NB\role\msg\NX.\PP.$ 
From $\NA\notin\fn{\receive\NB\role\msg\NX\PP},$ without loss of generality, we can conclude $\NA\not=\NB$ and $\NA\not=\NX.$ 
By Lemma \ref{lemm:Inv.Lemma} we get $ \Delta,\NA: \{\NA\}(T),\NX: T' \vdash_{\rho} \PP$, $\Delta(\NB)=\omega(T'),$ and if $\NB\notin\rho$ then $\omega \subseteq\rho.$ 
By induction hypothesis $ \Delta',\NX: T'\subst{\NR}{\NA} \vdash_{\rho} \PP.$
Since $\NA \notin\rho$ then if $\NA\in \omega$ it follows $\NB\in \rho.$ If $\NA\notin \omega$ then $\omega\subst{\NR}{\NA}=\omega.$ Thus, using \rulename{(t-in)} we get $ \Delta' \vdash_\rho \receive\NB\role\msg\NX.\PP.$

-If $\NA \leftrightarrow \nub$ and $ \Delta,\NA: \{\NA\} (T) \vdash_\rho \receive\NB\role\msg\NX.\PP.$ 
Using the same arguments we can again assume $\NA\not=\NB$ and $\NA\not=\NX.$ 
By Lemma \ref{lemm:Inv.Lemma} we get $ \Delta,\NA: \nub(T),\NX: T' \vdash_{\rho} \PP$, $\Delta(\NB)=\omega(T'),$ and if $\NB\notin\rho$ then $\omega \subseteq\rho.$ 
By induction hypothesis $ \Delta,\NX: T' \vdash_{\rho} \PP.$
Thus, using \rulename{(t-in)} we get $ \Delta \vdash_\rho \receive\NB\role\msg\NX.\PP.$

\end{proof}

\begin{lemma}[Subject Congruence]{\label{lemm:Subject_congruence}}
If $ \Delta\vdash_\rho \PP$ and $\PP\equiv \PQ$ then $ \Delta\vdash_\rho \PQ.$
%\htv{what about $\inact \equiv \rest{\NA: \_} \inact$ ($\Rightarrow$)? I would also add the rule for replication
%(folding direction)}
\end{lemma}
\begin{proof}
The proof is by induction on the depth of the derivation $\PP\equiv \PQ.$ We comment only three cases when the last applied rule is \rulename{(sc-par-inact)}, \rulename{(sc-res-extr)} or \rulename{(sc-rep)}:
\begin{itemize}
\item[$\mathit{1.}$] $\PP\parop\inact\equiv\PP$. From $ \Delta\vdash_\rho \PP\parop\inact$ by Lemma~\ref{lemm:Inv.Lemma} we get $ \Delta\vdash_{\rho_1} \PP $ and $ \Delta\vdash_{\rho_2} \inact,$ where $\rho_1 \uplus\rho_2=\rho.$ By Lemma~\ref{lemm:Weakening_Lemma} we get $ \Delta\vdash_\rho \PP.$

If $ \Delta\vdash_\rho \PP,$ by \rulename{(t-stop)} we get $ \Delta\vdash_\emptyset \inact$ and by \rulename{(t-par)} we conclude $ \Delta\vdash_\rho \PP\parop\inact.$

\item[$\mathit{2.}$]$\PP\parop \rest{\NA: \NR}\PQ \equiv \rest{\NA: \NR}(\PP\parop\PQ)$ or $\PP\parop \rest{\NA: \nub}\PQ \equiv \rest{\NA: \nub}(\PP\parop\PQ),$ if $\NA\notin\fn\PP.$ 

To show implication from right to the left we have two cases:

-If $\NA \leftrightarrow \NR$ then from $ \Delta\vdash_\rho \PP\parop \rest{\NA: \NR}\PQ$ by Lemma~\ref{lemm:Inv.Lemma} we get $ \Delta\vdash_{\rho_1} \PP$ and $ \Delta\vdash_{\rho_2}  \rest{\NA: \NR}\PQ,$ where $\rho_1\uplus\rho_2=\rho$ and $\sym{\PP} \cap \sym{\rest{\NA: \NR}\PQ}=\emptyset.$  
Applying Lemma~\ref{lemm:Inv.Lemma} we get $ \Delta',\NA:\{\NA\}(\varphi)\vdash_{\rho_2} \PQ$ where $\Delta'=\Delta\subst{\NA}{\NR}$ and $\NA\notin\rho_2$ and $\NR\notin\sym{\PQ}.$ From $\sym{\PP} \cap \sym{\rest{\NA: \NR}\PQ}=\emptyset$ we get $\NR\notin\sym{\PP}$ and from $\NA\notin\fn\PP$ and $\NA\in\bn{\rest{\NA : \NR}\PQ}$ without loss of generality can assume $\NA\notin\rho_1.$ By Lemma~\ref{lemm:Weakening_Lemma} we get $ \Delta',\NA:\{\NA\}(\varphi)\vdash_{\rho_1}\PP$ where $\Delta'=\Delta\subst{\NA}{\NR}.$ 
By \rulename{(t-par)} $ \Delta',\NA:\{\NA\}(\varphi)\vdash_\rho \PP\parop\PQ$ and by \rulename{(t-new)} $ \Delta\vdash_\rho \rest{\NA: \NR}(\PP\parop\PQ).$

-If $\NA \leftrightarrow \nub$ then from $ \Delta\vdash_\rho \PP\parop \rest{\NA: \nub}\PQ$ by Lemma~\ref{lemm:Inv.Lemma} we get $ \Delta\vdash_{\rho_1} \PP$ and $ \Delta\vdash_{\rho_2}  \rest{\NA: \nub}\PQ,$ where $\rho_1\uplus\rho_2=\rho.$ 
Applying Lemma~\ref{lemm:Inv.Lemma} again we get $ \Delta,\NA:\nub(\varphi)\vdash_{\rho_2} \PQ$ where $\NA\notin\rho_2.$ 
By Lemma~\ref{lemm:Weakening_Lemma} we get $ \Delta,\NA:\nub(\varphi)\vdash_{\rho_1}\PP.$ 
By \rulename{(t-par)} $ \Delta,\NA:\nub(\varphi)\vdash_\rho \PP\parop\PQ$ and by \rulename{(t-new-rep)} $ \Delta\vdash_\rho \rest{\NA: \nub}(\PP\parop\PQ).$

To show implication from left to the right we again have two cases:

-If $\NA \leftrightarrow \NR$ then from $ \Delta\vdash_\rho \rest{\NA: \NR} ( \PP\parop \PQ)$ by Lemma~\ref{lemm:Inv.Lemma} we get $ \Delta', \NA: \NA(T)\vdash_{\rho} \PP \parop \PQ$ where $\Delta'=\Delta\subst{\NA}{\NR}$ and $\NR \notin\sym{\PP\parop \PQ}$ and $\NA\notin \rho.$ By Lemma~\ref{lemm:Inv.Lemma} 
$ \Delta', \NA:\{\NA\}(T)\vdash_{\rho_1} \PP$ and $ \Delta', \NA: \{\NA\}(T)\vdash_{\rho_2}  \PQ,$ where $\rho_1\uplus\rho_2=\rho$ and $\sym{\PP} \cap \sym{\PQ}=\emptyset.$  Since $\NA\notin\fn{\PP}\cup \rho_1$ and $\NR\notin \sym{\PP}$ by Lemma~\ref{lemm:Strength.Lemmma} $ \Delta \vdash_{\rho_1} \PP.$ 
Using $\NR\notin\sym{\PQ}$ and $\NA\notin\rho_2$ by \rulename{(t-new)} we get $ \Delta \vdash_{\rho_2}  \rest{\NA: \NR}\PQ,$ and by \rulename{(t-par)} we get $ \Delta \vdash_\rho \PP \parop \rest{\NA: \NR} \PQ.$

-If $\NA \leftrightarrow \nub$ then from $ \Delta\vdash_\rho \rest{\NA: \nub} ( \PP\parop \PQ)$ by Lemma~\ref{lemm:Inv.Lemma} we get $ \Delta, \NA: \nub(T)\vdash_{\rho} \PP \parop \PQ$ where $\NA\notin \rho.$ By Lemma~\ref{lemm:Inv.Lemma} 
$ \Delta, \NA: \{\NA\}(T)\vdash_{\rho_1} \PP$ and $ \Delta, \NA: \{\NA\}(T)\vdash_{\rho_2}  \PQ,$ where $\rho_1\uplus\rho_2=\rho.$ Since $\NA\notin\fn{\PP}\cup \rho_1$ by Lemma~\ref{lemm:Strength.Lemmma} $ \Delta \vdash_{\rho_1} \PP.$ 
Using \rulename{(t-new-rep)} we get $ \Delta \vdash_{\rho_2}  \rest{\NA: \nub}\PQ,$ and by \rulename{(t-par)} we get $ \Delta \vdash_\rho \PP \parop \rest{\NA: \nub} \PQ.$

\item[$\mathit{3.}$] $\repreceive\NA\role\msg\NX.\PP \equiv \repreceive\NA\role\msg\NX.\PP \parop \scope\NA\receive\NA\role\msg\NX.\PP.$ We show only one implication. Suppose $ \Delta\vdash_\rho \repreceive\NA\role\msg\NX.\PP \parop \scope\NA\receive\NA\role\msg\NX.\PP.$ By Lemma~\ref{lemm:Inv.Lemma} we get $ \Delta\vdash_{\rho_1} \repreceive\NA\role\msg\NX.\PP$ and $ \Delta\vdash_{\rho_2} \scope\NA\receive\NA\role\msg\NX.\PP,$ where $\rho_1\uplus\rho_2=\rho.$ By the same Lemma we get $\sym{\PP}=\emptyset$ and $\Delta, \NX: T\vdash_{\{\NA\}} \PP,$ where $\Delta(\NA)=\ttype(T)$ and $\NX \notin \rho_1, \mathit{names}(\Delta).$ By \rulename{(t-rep-in)} we get $ \Delta\vdash_{\rho} \repreceive\NA\role\msg\NX.\PP.$

\end{itemize}
\end{proof}

%\comment{H: why not stating $x \not \in \Delta, \rho$ from starters? we have this in the input.. then no need to replace in $\rho$ right?}
\begin{lemmata}{\ref{lemm:Substitution_lemma}~(Substitution).}

Let $ \Delta,\NX:\ttype(T)\vdash_\rho \PP$ and $\NX\notin\mathit{names}(\Delta)$. Then 
\begin{itemize}
\item[1.] If $\Delta(\NA)=\{\NA\}(T)$  and $\NA\in \ttype$ then $\Delta\vdash_{\rho\subst{\NA}{\NX}} \PP\subst{\NA}{\NX}.$
\item[2.] If $\Delta(\NA)=\nub(T)$  and $\nub=\ttype$ then $\Delta\vdash_{\rho\subst{\NA}{\NX}} \PP\subst{\NA}{\NX}.$
\end{itemize}

\end{lemmata}

\begin{proof}
The proof is by induction on the depth of the derivation $ \Delta\vdash_\rho\PP.$ We detail two cases:
\begin{itemize}
\item  $ \Delta,\NX:\ttype(\ttype'(T))\vdash_\rho \send\NX\role\msg\NB.\PP.$ 
By Lemma \ref{lemm:Inv.Lemma} $ \Delta, \NX:\ttype(\ttype'(T))\vdash_{\rho}\PP,$ where $\Delta(\NB)=\ttype''(T)$ and $\ttype''\subseteq\ttype'$ and if $\NX\notin\rho$ then $\ttype\subseteq\rho.$ 
By induction hypothesis $ \Delta\vdash_{\rho\subst{\NA}{\NX}} \PP\subst{\NA}{\NX}.$ 
If $\Delta(\NA)=\{\NA\}(\omega'(T))$ and $\NX\in\rho$ then $\NA\in\rho\subst{\NA}{\NX}$ and if $\NX\not\in\rho$ then by $\ttype\subseteq\rho$ we get $\NA\in\ttype\subseteq\rho=\rho\subst{\NA}{\NX}.$ 
If $\Delta(\NA)=\nub(\omega'(T))$ then $\NX\in\rho$ which implies $\NA\in\rho\subst{\NA}{\NX}.$
By \rulename{(t-out)} we get $ \Delta\vdash_{\rho\subst{\NA}{\NX}} (\send\NX\role\msg\NB.\PP)\subst{\NA}{\NX}.$

\item $ \Delta,\NX:\ttype(T)\vdash_\rho \send\NB\role\msg\NX.\PP.$ 
By Lemma \ref{lemm:Inv.Lemma} $ \Delta, \NX:\ttype(T)\vdash_{\rho}\PP,$ where $\Delta(\NB)=\ttype'(\ttype''(T))$ and $\ttype\subseteq\ttype''$ and if $\NB\notin\rho$ then $\ttype\subseteq\rho.$ 
By induction hypothesis $ \Delta\vdash_{\rho\subst{\NA}{\NX}} \PP\subst{\NA}{\NX}.$ If $\Delta(\NA)=\{\NA\}(T)$ then $\NA\in\ttype\subseteq\ttype''$ and if $\NA: \nub(T)$ then $\nub=\ttype\subseteq\ttype''.$
Thus, by \rulename{(t-out)} we get $ \Delta\vdash_{\rho\subst{\NA}{\NX}} (\send\NB\role\msg\NX.\PP)\subst{\NA}{\NX}.$
\end{itemize}
\end{proof}

%\comment{H: please merge this result with the last Lemma, you already have subject congruence.}

\begin{lemma}[Authorization Safety]
\label{lem:authpresent}
If $\Delta \vdash_\emptyset \context[\PP_1,\PP_2]$ and 
$\Delta \vdash_{\rho_1} \PP_1$ and $\Delta \vdash_{\rho_2} \PP_2$
and $\NA\in\rho_1 \cap \rho_2$ then $ \optop{\context[\cdot_1,\cdot_2]}{\NA}{\NA}$ is defined.
\end{lemma}
\begin{proof}
By induction on the structure of $\context[\cdot_1,\cdot_2]$.
\end{proof}

\paragraph{\bf Notation} We use $\rest{\tilde\NC: \tilde{\Omega}}$ to abbreviate $\rest{\NC_1: \Omega_1}\ldots\rest{\NC_n: \Omega_n}$, where $\Omega$ ranges over symbols from $\cal{S}$
and $\nu$.
\begin{lemma}[Interaction Safety]{\label{lemm:Error_Free}}
\begin{itemize}
%\item[]{}
\item[1.] If $\PP\equiv \rest{\tilde\NC: \tilde{\Omega}}\context[\send\NA\role\msg\NB.\PP_1, \receive\NA\role\msg\NX.\PP_2]$ and $\PP$ is well-typed with $\Delta\vdash_\emptyset \PP$ then $ \optop{\context[\cdot_1,\cdot_2]}{\NA}{\NA}$ is defined and if $\context'[\cdot_1, \cdot_2]= \optop{\context[\cdot_1,\cdot_2]}{\NA}{\NA}$ and $\PQ\equiv \rest{\tilde\NC: \tilde{\Omega}}\context'[\scope\NA\PP_1, \scope\NA\PP_2\subst{\NB}{\NX}]$ then $ \Delta\vdash_\emptyset \PQ.$
\item[2.] If $\PP\equiv \rest{\tilde\NC: \tilde{\Omega}}\context[\sauth\NA\role\msg\NB.\PP_1, \rauth\NA\role\msg\NB.\PP_2]$ and $\PP$ is well-typed with $\Delta\vdash_\emptyset \PP$ then $ \optop{\context[\cdot_1,\cdot_2]}{\NA, \NB}{\NA}$ is defined and if $\context'[\cdot_1, \cdot_2]= \optop{\context[\cdot_1,\cdot_2]}{\NA, \NB}{\NA}$ and $\PQ\equiv \rest{\tilde\NC: \tilde{\Omega}}\context'[\scope\NA\PP_1, \scope\NA\scope\NB\PP_2]$ then $ \Delta\vdash_\emptyset \PQ.$
\end{itemize}

\end{lemma}

\begin{proof}
The proof is by induction on the structure of the context $\context[\cdot_1, \cdot_2].$ We detail only the first statement. If $ \Delta\vdash_\emptyset \PP$ by Lemma~\ref{lemm:Subject_congruence} we get  $ \Delta\vdash_\emptyset \rest{\tilde\NC: \tilde{\Omega}}\context[\send\NA\role\msg\NB.\PP_1, \receive\NA\role\msg\NX.\PP_2]$ and by consecutive application of Lemma ~\ref{lemm:Inv.Lemma}. 1 and 2, we get $ \Delta'\vdash_\emptyset \context[\send\NA\role\msg\NB.\PP_1, \receive\NA\role\msg\NX.\PP_2],$
where $\Delta'= \Delta, \Delta''$ and for each $\NC\in\mathit{dom}(\Delta''),$ $\Delta''(\NC)=\NC(T)$ or $\Delta''(\NC)=\nub(T).$ By consecutive application of Lemma~\ref{lemm:Inv.Lemma}. 3 and 9, we get 
\[
 \Delta'\vdash_{\rho_1} \send\NA\role\msg\NB.\PP_1 \quad \text{and} \quad  \Delta'\vdash_{\rho_2} \receive\NA\role\msg\NX.\PP_2,
\]
for some multisets of names $\rho_1, \rho_2.$ 
By the same Lemma again 
\[
 \Delta'\vdash_{\rho_1} \PP_1 \quad \text{and} \quad  \Delta', \NX:\ttype(T)\vdash_{\rho_2}\PP_2,
\]
where $\Delta'(\NA)=\{\NA\}(\ttype(T))$ or $\Delta'(\NA)=\nub(\ttype(T))$ and thus $\NA\in\rho_1,$ $\NA\in\rho_2.$  Furthermore, $\Delta'(\NB)=\{\NB\}(T)$ or $\Delta'(\NB)=\nub(T),$ and $\NB\in\ttype$  and $\NX\notin\rho_2\cup\mathit{names}(\Delta')$. Hence we have $\rho_2\subst{\NB}{\NX}=\rho_2,$ and by Lemma \ref{lemm:Substitution_lemma} we get $\Delta'\vdash_{\rho_2}\PP_2\subst{\NB}{\NX}.$
By \rulename{(t-auth)} we get 
\[
 \Delta'\vdash_{\rho_1'} \scope\NA \PP_1 \quad \text{and} \quad  \Delta'\vdash_{\rho_2'} \scope \NA\PP_2\subst{\NB}{\NX},
\]
where $\rho_1=\rho_1'\uplus\{\NA\}$ and $\rho_2=\rho_2'\uplus\{\NA\}.$

Since $ \Delta'\vdash_\emptyset \context[\send\NA\role\msg\NB.\PP_1, \receive\NA\role\msg\NX.\PP_2],$ and $\NA\in\rho_1,$ $\NA\in\rho_2$ by Lemma~\ref{lem:authpresent} we conclude
 $\optop{\context[\cdot_1,\cdot_2]}{\NA}{\NA}$ is defined.
By Proposition~\ref{prop:Cases_for_contexts} we distinguish four cases for the structure of the context $\context[\cdot_1, \cdot_2]=\context''[\context_1[\cdot_1] \parop \context_2[\cdot_2]].$ 
We comment only the case when $\context_1[\cdot]=\context_1'[\scope\NA \context_1''[\cdot]]$ and $\context_2[\cdot]=\context_2'[\scope\NA \context_2''[\cdot]],$ where $\operator(\context_1''[\cdot]; \NA)$ and $\operator(\context_2''[\cdot]; \NA)$ are undefined, thus in contexts $\context_1''[\cdot]$ and $\context_2''[\cdot]$ the holes are not in the scope of authorizations $\scope\NA.$  
By consecutive application of \rulename{(t-par)} and \rulename{(t-auth)} we get $ \Delta' \vdash_{\rho_1''} \scope\NA \context_1''[\send\NA\role\msg\NB.\PP_1]$ and $ \Delta' \vdash_{\rho_2''} \scope\NA \context_2''[\receive\NA\role\msg\NX.\PP_2],$ but also $ \Delta' \vdash_{\rho_1''}  \context_1''[\scope\NA\PP_1]$ and $ \Delta' \vdash_{\rho_2''}  \context_2''[\scope\NA\PP_2\subst{\NB}{\NX}],$ for some $ \rho_1'', \rho_2''.$ 
Since $\context'[\cdot_1, \cdot_2]=\context''[\context_1'[\context_1''[\cdot_1]] \parop \context_2'[\context_2''[\cdot_2]]]$, again by consecutive application of \rulename{(t-par)} and \rulename{(t-auth)} we get $\Delta'\vdash_\emptyset \context'[\scope\NA\PP_1, \scope\NA\PP_2\subst{\NB}{\NX}]$ then by consecutive application of \rulename{(t-new)} and \rulename{(t-new-rep)} we get $\Delta\vdash_\emptyset \rest{\tilde\NC: \tilde{\Omega}}\context'[\scope\NA\PP_1, \scope\NA\PP_2\subst{\NB}{\NX}]$.
By Lemma~\ref{lemm:Subject_congruence} we get $\Delta\vdash_\emptyset \PQ.$
\end{proof}

\begin{theorema}{\ref{theorem:Subject_reduction}~(Subject Reduction).}
If $\PP$ is well-typed, $ \Delta\vdash_\emptyset \PP$, and $\PP\red\PQ$ then $ \Delta\vdash_\emptyset \PQ.$

\end{theorema}

\begin{proof}
The proof is by case analysis on last reduction step. 
We have two base cases, when the rules applied are \rulename{(r-comm)} or \rulename{(r-auth)}, which both we get directly by Lemma~\ref{lemm:Error_Free}. 
For the induction steps we have two cases:
\begin{itemize}
\item If the last applied rule is \rulename{(r-newc)} then we have two cases
	\begin{itemize}
	\item $\rest{\NA: \NR} \PP' \red \rest{\NA: \NR} \PQ'$ is derived from $ \PP' \red  \PQ'.$ Let $ \Delta\vdash_\emptyset \rest{\NA: \NR} \PP'.$ By Proposition~\ref{lemm:well-formed} we get $\rest{\NA: \NR} \PQ'$ is well-formed, thus $\NR\notin\sym{\PQ'}$.
By Lemma~\ref{lemm:Inv.Lemma} we get $ \Delta', \NA: \{\NA\}(T) \vdash_\rho \PP',$ where  $\NR\notin\sym{\PP'}$ and $\Delta'=\Delta\subst{\NA}{\NR}.$ By induction hypothesis we get $ \Delta', \NA: \{\NA\}(T) \vdash_\rho \PQ',$ and by \rulename{(t-new)} we get $ \Delta\vdash_\emptyset \rest{\NA: \NR} \PQ'.$ 
	\item $\rest{\NA: \nub} \PP' \red \rest{\NA: \nub} \PQ'$ is derived from $ \PP' \red  \PQ'.$ Let $ \Delta\vdash_\emptyset \rest{\NA: \nub} \PP'.$ By Lemma~\ref{lemm:Inv.Lemma} we get $ \Delta, \NA: \nub(T) \vdash_\rho \PP'.$  By induction hypothesis we get $ \Delta, \NA: \nub(T) \vdash_\rho \PQ',$ and by \rulename{(t-new-rep)} we get $ \Delta\vdash_\emptyset \rest{\NA: \nub} \PQ'.$
	\end{itemize} 
\item If the last applied rule is \rulename{(r-struc)} then $ \PP \red \PQ$ is derived from $ \PP' \red  \PQ',$ where $\PP\equiv \PP'$ and $\PQ\equiv \PQ'.$ Let $ \Delta\vdash_\emptyset \PP.$ By Lemma~\ref{lemm:Subject_congruence} we get $ \Delta\vdash_\emptyset \PP'.$ By induction hypothesis $ \Delta\vdash_\emptyset \PQ'$ and by Lemma~\ref{lemm:Subject_congruence} we get $ \Delta\vdash_\emptyset \PQ.$
\end{itemize}
\end{proof}

\begin{prop}{\ref{lemm:Error_Freedom}~(Typing Soundness)}

\begin{itemize}
\item[]{}
\item[$\mathit{1.}$] If $\Delta\vdash_\rho \PP$ then $\PP$ is well-formed.
\item[$\mathit{2.}$] If $\PP$ is well-typed then $\PP$ is not an error.
\end{itemize}
\end{prop}
\begin{proof}
\begin{itemize}
\item[]{}
\item[$\mathit{1.}$] Directly from the typing rules.
\item[$\mathit{2.}$] Immediate from Lemma~\ref{lemm:Error_Free}. 
\end{itemize}
\end{proof}

%\comment{H: Please merge with other Error absence}
%\begin{lemmata}{\ref{lemm:Error_Free}~(Error Absence).}

%If $ \Delta\vdash_\emptyset \PP$ and for each $\NA\in\mathit{dom}(\Delta),$ $\Delta(\NA)=\NA(T)$ or $\Delta(\NA)=\nub(T)$ and $P\equiv \rest{\tilde\NC} \context[\alpha_\NA.\PQ, \alpha_\NA'.\PR]$ where 
%\begin{enumerate}
%\item 
%$\alpha_\NA = \send\NA\role\msg\NB$, 
%$\alpha_\NB = \receive\NA\role\msg\NX$
%then
%$ \optop{\context[\cdot_1,\cdot_2]}{\NA}{\NA}$  is defined, or 
%\item 
%$\alpha_\NA = \sauth\NA\role\msg\NB$, 
%$\alpha_\NB = \rauth\NA\role\msg\NB$
%then
%$ \optop{\context[\cdot_1,\cdot_2]}{\NA,\NB}{\NA}$ is defined.
%\end{enumerate}
%Therefore, $\PP$ is not an error.
%\end{lemmata}

%\begin{proof}
%Directly from Lemma~\ref{lemm:Pre_Subject_reduction} and Lemma~\ref{lemm:Subject_congruence}.
%\end{proof}

\begin{corollar}{\ref{cor:Type_Safety}~(Type Safety).}

If $\PP$ is well-typed and $\PP\red^*\PQ$ then $\PQ$ is not an error.
\end{corollar}
\begin{proof}
The proof follows directly from Theorem~\ref{theorem:Subject_reduction} and Proposition~\ref{lemm:Error_Freedom}(2).
\end{proof}

\end{document}